\documentclass[journal]{IEEEtran}

\usepackage{
amssymb,
latexsym,
soul,
mathrsfs,
times,
blkarray,
enumitem,
multirow
}
\usepackage[cmex10]{amsmath}
\usepackage{amsthm}
\usepackage[english]{babel}
\usepackage[utf8]{inputenc}
\usepackage[T1]{fontenc}
\usepackage[colorinlistoftodos]{todonotes}
\usepackage{csquotes}

\newtheorem{theorem}{Theorem}[section]
\newtheorem{proposition}[theorem]{Proposition}
\newtheorem{lemma}[theorem]{Lemma}
\newtheorem{corollary}[theorem]{Corollary}
\newtheorem{conjecture}[theorem]{Conjecture}
\newtheorem{claim}[theorem]{Claim}

\theoremstyle{definition}
\newtheorem{algorithm}[theorem]{Algorithm}
\newtheorem{example}[theorem]{Example}
\newtheorem{definition}[theorem]{Definition}
\newtheorem{remark}[theorem]{Remark}

\numberwithin{equation}{section} \numberwithin{figure}{section}
\numberwithin{table}{section}

\newcommand{\nc}{\newcommand}

\makeatletter
\newcommand*\Wt[1]{\mathpalette\wthelper{#1}}
\newcommand*\wthelper[2]{%
\hbox{\dimen@\accentfontxheight#1%
\accentfontxheight#11.15\dimen@
$\m@th#1\widetilde{#2}$%
\accentfontxheight#1\dimen@
}%
}

\newcommand*\accentfontxheight[1]{%
\fontdimen5\ifx#1\displaystyle
\textfont
\else\ifx#1\textstyle
\textfont
\else\ifx#1\scriptstyle
\scriptfont
\else
\scriptscriptfont
\fi\fi\fi3
}
\makeatother

\nc{\Z}{\mathbb{Z}}
\nc{\C}{\mathcal{C}}
\nc{\tC}{\Wt{\C}}
\nc{\tG}{\widetilde{G}}
\nc{\oC}{\overline{C}}
\nc{\oG}{\overline{G}}
\nc{\hC}{\widehat{\C}}
\nc{\bx}{\mathbf{x}}
\nc{\by}{\mathbf{y}}
\nc{\wt}{\mathrm{wt}}
\nc{\SO}{\mathrm{SO}}
\nc{\sfh}{\mathsf{h}}
\nc{\nmr}{\mathnormal{r}}
\nc{\nmc}{\mathnormal{c}}
\nc{\calI}{\mathcal{I}}
\nc{\calP}{\mathcal{P}}
\nc{\calS}{\mathcal{S}}
\nc{\nmH}{H}
\nc{\calE}{\mathcal{E}}
\nc{\bfI}{\mathbf{I}}
\nc{\hG}{\widehat{G}}
\nc{\fn}{\mathfrak{n}}
\nc{\hfn}{\widehat{\fn}}
\nc{\bfv}{\mathbf{v}}
\nc{\hs}{\widehat{s}}
\nc{\wcI}{\widehat{\calI}}

\nc{\dso}{d_{\mathrm{so}}}

\nc{\floor}[1]{\left\lfloor #1 \right\rfloor}
\nc{\ceil}[1]{\left\lceil #1 \right\rceil}

\nc{\yh}[1]{\todo[size=\tiny,color=blue!10]{#1 \\ \hfill --- Young-Hun}}
\nc{\YH}[1]{\todo[size=\tiny,inline,color=blue!10]{#1
\\ \hfill --- Young-Hun}}
\nc{\N}[1]{\todo[size=\tiny,inline,color=pink!30]{#1
\\ \hfill --- Nari}}

\newenvironment{red}{\relax\color{red}}{\relax}
\newenvironment{blue}{\relax\color{blue}}{\hspace*{.5ex}\relax}
\newenvironment{teal}{\relax\color{teal}}{\hspace*{.5ex}\relax}
\newenvironment{magenta}{\relax\color{magenta}}{\hspace*{.5ex}\relax}

\nc{\ber}{\begin{red}}
\nc{\er}{\end{red}}
\nc{\beb}{\begin{blue}}
\nc{\eb}{\end{blue}}
\nc{\bema}{\begin{magenta}}
\nc{\ema}{\end{magenta}}
\nc{\bete}{\begin{teal}}
\nc{\ete}{\end{teal}}

\begin{document}

\title{Embedding linear codes into self-orthogonal codes and their optimal minimum  distances}

\author{Jon-Lark~Kim,
Young-Hun~Kim,
and~Nari~Lee
\thanks{J.-L. Kim was supported by Basic Research Program through the National Research Foundation of
Korea (NRF) funded by the Ministry of Education (NRF-2019R1A2C1088676).}
\thanks{J.-L. Kim is with the Department of Mathematics, Sogang University, Seoul 04107, Republic of Korea (e-mail: jlkim@sogang.ac.kr).}
\thanks{Y.-H. Kim is with the Department of Mathematics, Sogang University, Seoul 04107, Republic of Korea, 
Research Institute for Basic Science, Sogang University, Seoul 04107, Republic of Korea, and
Department of Mathematics, Ewha Womans University, Seoul 03760, Republic of Korea (e-mail: ykim.math@gmail.com).}
\thanks{N. Lee is with the Affiliated Institute of ETRI, Daejeon 34044, Republic of Korea (e-mail: narilee@nsr.re.kr).}}

%
%

%

\maketitle

\begin{abstract}
We obtain a characterization on self-orthogonality for a given binary linear code in terms of the number of column vectors in its generator matrix, which extends the result of Bouyukliev et al. (2006). As an application, we give an algorithmic method to embed a given binary $k$-dimensional linear code $\C$ ($k = 2,3,4$) into a self-orthogonal code of the shortest length which has the same dimension $k$ and minimum distance $d' \ge d(\C)$.  For $k > 4$, we suggest a recursive method to embed a $k$-dimensional linear code to a self-orthogonal code. 
We also give new explicit formulas for the minimum distances of optimal self-orthogonal codes for any length $n$ with dimension 4 and any length $n \not\equiv 6,13,14,21,22,28,29 \pmod{31}$ with dimension 5. We determine the exact optimal minimum distances of $[n,4]$ self-orthogonal codes which were left open by Li-Xu-Zhao (2008) when $n \equiv 0,3,4,5,10,11,12 \pmod{15}$.  Then, using MAGMA, we observe that our embedding sends an optimal linear code to an optimal self-orthogonal code.
\end{abstract}


\section{Introduction}
\IEEEPARstart{S}{elf}-orthogonal codes have been extensively studied for their interesting structures and applications. In particular, self-dual codes, a special class of self-orthogonal codes, have attracted much attention because of their connections to other fields of mathematics such as unimodular lattices, secret sharing schemes, and designs (\cite{07CCGHV, 08DMS, 09Ha, 13BBH, 04BG, 04HKM}). 

Since the 1970s, a lot of researchers have studied self-orthogonal codes. For instance, constructions and classification of self-orthogonal codes were steadily studied (\cite{72Pless,75PleSlo,04Bouyu,08Li,19CES}). Self-orthogonal codes were also studied due to their connections to quantum codes (\cite{02Kim,10JLLX,12JX,14LV,17JKW}) and their applications to the resistance to side-channel attacks (\cite{13CG, 13CRZ}). 
However, several questions for self-orthogonal codes remain. To mention a few, the classification of self-orthogonal codes and the explicit formulas for the minimum distances of optimal self-orthogonal codes are partially computed.

From now on, we will only consider binary linear codes.
Pless~\cite{72Pless} classified certain self-orthogonal codes. Since then people have gotten more results for the classification of self-orthogonal codes.
In~\cite[Section 3]{06bou}, Bouyukliev et al. introduced a noteworthy characterization for self-orthogonality of three-dimensional codes in terms of the number of column vectors in a generator matrix. As a consequence, they gave the complete classification of three-dimensional optimal self-orthogonal codes.
In~\cite{08Li}, Li, Xu, and Zhao characterized four-dimensional optimal self-orthogonal codes by systems of linear equations. They also obtained the complete classification of optimal $[n,4]$ self-orthogonal codes for $n \equiv 1,2,6,7,8,9,13,14 \pmod{15}$ and left the other cases open.

In this paper, we generalize the characterization in~\cite[Section 3]{06bou} for arbitrary dimensions. In particular, we give an explicit characterization for $[n,k]$ self-orthogonal codes for $k = 2,4$ and reprove the characterization for $[n,3]$ self-orthogonal codes, which was introduced in \cite{06bou}.

As a consequence of our characterizations, we construct an algorithm that  embeds (or extends) an $[n,k]$ linear code to a self-orthogonal code for $k = 2,3,4$ in Section~\ref{sec: algorithm}. 
Precisely, if we input a generator matrix $G$ of a linear code $\C$, then the algorithm will give a matrix $\tG$ by adding more columns to $G$ which generates a self-orthogonal code $\tC$ and produces a minimum distance greater than or equal to the minimum distance of $\C$. Moreover, we prove that $\tC$ is a shortest (length) self-orthogonal embedding.
In~\cite{95Kob}, Kobayashi and Takada introduced a similar embedding with a critical error. We point out their error (see Remark~\ref{rem: error in 95Kob}).
For $k > 4$, we suggest a recursive method to embed a $k$-dimensional linear code to a self-orthogonal code.

It is a quite natural question whether an optimal linear code results in an optimal self-orthogonal code when it is embedded by our algorithm in Section~\ref{sec: algorithm}. Therefore, we will also discuss the explicit formulas for the minimum distances of optimal linear codes and optimal SO codes in Section~\ref{sec: min dist}.

There have been a number of studies on bounds about the minimum distances of linear codes. In particular, Griesmer~\cite{60Griesmer} introduced a remarkable bound, called the Griesmer bound. Due to this bound, researchers obtained a lot of minimum distances of optimal linear codes. For the details, readers can refer to~\cite{07codetables}. 

Manipulating the Griesmer bound, we calculate an explicit formula for an upper bound of the minimum distance $d(n,k)$ of optimal linear codes of length $n$ for $k = 1,2,3,4,5$. More precisely, we obtain new explicit formulas for the minimum distances $\dso(n,k)$ of optimal self-orthogonal codes for any length $n$ with $k=4$ and for any length $n \not\equiv 6,13,14,21,22,28,29 \pmod{31}$ with $k=5$. We determine the exact optimal minimum distances of $[n,4]$ self-orthogonal codes which were left open by Li-Xu-Zhao in~\cite{08Li} when $n \equiv 0,3,4,5,10,11,12 \pmod{15}$.  Using the code database in MAGMA~\cite{97MAGMA} and our explicit formulas, we observe that our embedding algorithm sends an optimal linear code to an optimal self-orthogonal code.

Readers may refer to~\cite{07codetables} to see the list of $d(n,k)$ for $n,k \le 256$ and refer to~\cite{06bou} to see the known $\dso(n,k)$ for $n\le 40$ and $k \le 10$. For $40 < n \le 100$, the upper bound of $\dso(n,5)$ is obtained from the Griesmer bound~\cite{60Griesmer}. In Tables~\ref{table: dso(n,k)1} and~\ref{table: dso(n,k)2}, we list $\dso(n,k)$ for $n \le 100$ and $k = 4,5$.
Here  the superscript $*$ and $\dagger$ respectively denote our exact value and the conjectured value from Conjecture~\ref{conj: dso(n,5)}.

This paper is organized as follows. In Section~\ref{sec: prelim}, we introduce basic facts. In Section~\ref{sec: equiv cond}, we give characterizations for self-orthogonal codes. In Section~\ref{sec: algorithm}, we introduce algorithms which give shortest self-orthogonal embeddings using our characterizations. In Section~\ref{sec: min dist}, we provide explicit formulas for $d(n,k)$ and $\dso(n,k)$ for $k = 1,2,3,4,5$.

\begin{table}[h]
\centering
\renewcommand{\arraystretch}{1.03}
\begin{tabular}{c||r|c||r|c}
\multirow{2}{*}{$n$}	& \multicolumn{2}{c||}{$\dso(n,4)$}	 &	\multicolumn{2}{c}{$\dso(n,5)$}\\ \cline{2-5}
&Known [ref]	  &Improved 	& Known [ref]	& Improved \\\hline\hline
$8  $ &		$4$		&		-		&				&					\\\hline
$9  $ &		$4$		&		-		&				&					\\\hline
$10 $ &		$4$		&		-		&	$4$			&		-		\\\hline
$11 $ &		$4$		&		-		&	$4$			&		-		\\\hline
$12 $ &		$4$		&		-		&	$4$			&		-		\\\hline
$13 $ &		$4$		&		-		&	$4$			&		-		\\\hline
$14 $ &		$6$		&		-		&	$4$			&		-		\\\hline
$15 $ &		$8$		&		-		&	$6$			&		-		\\\hline
$16 $ &		$8$		&		-		&	$8$			&		-		\\\hline
$17 $ &		$8$		&		-		&	$8$			&		-		\\\hline
$18 $ &		$8$		&		-		&	$8$			&		-		\\\hline
$19 $ &		$8$		&		-		&	$8$			&		-		\\\hline
$20 $ &		$8$		&		-		&	$8$			&		-		\\\hline
$21 $ &		$10$	&		-		&	$8$			&		-		\\\hline
$22 $ &		$10$	&		-		&	$8$			&		-		\\\hline
$23 $ &		$12$	&		-		&	$10$		&		-		\\\hline
$24 $ &		$12$	&		-		&	$12$		&		-		\\\hline
$25 $ &		$12$	&		-		&	$12$		&		-		\\\hline
$26 $ &		$12$	&		-		&	$12$		&		-		\\\hline
$27 $ &		$12$	&		-		&	$12$		&		-		\\\hline
$28 $ &		$14$	&		-		&	$12$		&		-		\\\hline
$29 $ &		$14$	&		-		&	$12$		&		-		\\\hline
$30 $ &		$16$	&		-		&	$14$		&		-		\\\hline
$31 $ &		$16$	&		-		&	$16$		&		-		\\\hline
$32 $ &		$16$	&		-		&	$16$		&		-		\\\hline
$33 $ &		$16$	&		-		&	$16$		&		-		\\\hline
$34 $ &		$16$	&		-		&	$16$		&		-		\\\hline
$35 $ &		$16$	&		-		&	$16$		&		-		\\\hline
$36 $ &		$18$	&		-		&	$16$		&		-		\\\hline
$37 $ &		$18$	&		-		&	$16$		&		-		\\\hline
$38 $ &		$20$	&		-		&	$18$		&		-		\\\hline
$39 $ &		$20$	&		-		&	$18$		&		-		\\\hline
$40 $ &		$20$	&		-		&	$20$		&		-		\\\hline
$41 $ &		$\le 20~\cite{60Griesmer}$	&		${\bf 20^*}$		&	$20 $		&		-		\\\hline
$42 $ &		$\le 20~\cite{08Li}$	    &		${\bf 20^*}$		&	$20 $		&		-		\\\hline
$43 $ &		$22~\cite{08Li}$	&		-		&	$\le 20~\cite{60Griesmer} $		&		${\bf 20^*}$		\\\hline
$44 $ &		$22~\cite{08Li}$	&		-		&	$\le 22~\cite{60Griesmer} $		&		$20^{\dagger}$	\\\hline
$45 $ &		$24~\cite{97MAGMA}$	&		-		&	$\le 22~\cite{60Griesmer} $		&		$20^{\dagger}$	\\\hline
$46 $ &		$24~\cite{08Li}$	&		-		&	$\le 22~\cite{60Griesmer} $		&		${\bf 22^*}$		\\\hline
$47 $ &		$24~\cite{08Li}$	&		-		&	$24~\cite{97MAGMA} $		&		-		\\\hline
$48 $ &		$24~\cite{97MAGMA}$	&		-		&	$24~\cite{97MAGMA} $		&		-		\\\hline
$49 $ &		$\le 24~\cite{60Griesmer}$	&		${\bf 24^*}$		&	$24~\cite{97MAGMA} $		&		-		\\\hline
$50 $ &		$\le 24~\cite{08Li}$	&		${\bf 24^*}$		&	$24~\cite{97MAGMA} $		&		-	
\end{tabular}
\caption{$\dso(n,k)$ for $n \le 50$ and $k=4,5$}
\label{table: dso(n,k)1}
\renewcommand{\arraystretch}{1}
\end{table}
\newpage

\begin{table}[h]
\centering
\renewcommand{\arraystretch}{1.03}
\begin{tabular}{c||r|c||r|c}
\multirow{2}{*}{$n$} &   \multicolumn{2}{c||}{$\dso(n,4)$}   &     \multicolumn{2}{c}{$\dso(n,5)$}      \\ \cline{2-5}
&                Known	[ref] & Improved & Known [ref] &    Improved    \\ \hline\hline
$51 $         &           $26~\cite{08Li}$ &   -   &        $\le 24~\cite{60Griesmer}$        &     ${\bf 24^*}$     \\ \hline
$52 $         &           $26~\cite{08Li}$ &   -   &        $\le 26~\cite{60Griesmer} $        & $24^{\dagger}$ \\ \hline
$53 $         &           $28~\cite{08Li}$ &   -   &        $\le 26~\cite{60Griesmer} $        & $24^{\dagger}$ \\ \hline
$54 $         &           $28~\cite{08Li}$ &   -   &        $\le 26~\cite{60Griesmer} $        &     ${\bf 26^*}$     \\ \hline
$55 $         & $28~\cite{97MAGMA}$ &  -  &        $28~\cite{97MAGMA} $        &     -     \\ \hline
$56 $         & $\le 28~\cite{60Griesmer}$ &  ${\bf 28^*}$  &        $28~\cite{97MAGMA} $        &     -     \\ \hline
$57 $         &       $\le 28~\cite{08Li}$ &  ${\bf 28^*}$  &        $28~\cite{97MAGMA} $        &     -     \\ \hline
$58 $         &           $30~\cite{08Li}$ &   -   &        $\le 28~\cite{60Griesmer} $        &     ${\bf 28^*}$     \\ \hline
$59 $         &           $30~\cite{08Li}$ &   -   &        $\le 30~\cite{60Griesmer} $        & $28^{\dagger}$ \\ \hline
$60 $         & $32~\cite{97MAGMA}$ &  -  &        $\le 30~\cite{60Griesmer} $        & $28^{\dagger}$ \\ \hline
$61 $         &           $32~\cite{08Li}$ &   -   &        $\le 30~\cite{60Griesmer} $        &     ${\bf 30^*}$     \\ \hline
$62 $         &           $32~\cite{08Li}$ &   -   &        $32~\cite{97MAGMA}$        &     -     \\ \hline
$63 $         & $32~\cite{97MAGMA}$ &  -  &        $32~\cite{97MAGMA} $        &     -     \\ \hline
$64 $         & $\le 32~\cite{60Griesmer}$ &  ${\bf 32^*}$  &        $32~\cite{97MAGMA} $        &     -     \\ \hline
$65 $         &       $\le 32~\cite{08Li}$ &  ${\bf 32^*}$  &        $32~\cite{97MAGMA} $        &     -     \\ \hline
$66 $         &           $34~\cite{08Li}$ &   -   &        $32~\cite{97MAGMA} $        &     -     \\ \hline
$67 $         &           $34~\cite{08Li}$ &   -   &        $\le 32~\cite{60Griesmer} $        &     ${\bf 32^*}$     \\ \hline
$68 $         &           $36~\cite{08Li}$ &   -   &        $\le 34~\cite{60Griesmer} $        & $32^{\dagger}$ \\ \hline
$69 $         &           $36~\cite{08Li}$ &   -   &        $\le 34~\cite{60Griesmer} $        &     ${\bf 34^*}$     \\ \hline
$70 $         & $36~\cite{97MAGMA}$ &  -  &        $\le 34~\cite{60Griesmer} $        &     ${\bf 34^*}$ \\ \hline
$71 $         & $\le 36~\cite{60Griesmer}$ &  ${\bf 36^*}$  &        $ 36~\cite{97MAGMA} $        &     -     \\ \hline
$72 $         &       $\le 36~\cite{08Li}$ &  ${\bf 36^*}$   &        $ 36~\cite{97MAGMA} $        &    -     \\ \hline
$73 $         &           $38~\cite{08Li}$ &   -   &        $\le 36~\cite{60Griesmer} $        &     ${\bf 36^*}$     \\ \hline
$74 $         &           $38~\cite{08Li}$ &   -   &        $\le 36~\cite{60Griesmer} $        &     ${\bf 36^*}$     \\ \hline
$75 $         & $40~\cite{97MAGMA}$ &  -  &        $\le 38~\cite{60Griesmer} $        & $36^{\dagger}$ \\ \hline
$76 $         &           $40~\cite{08Li}$ &  -   &        $\le 38~\cite{60Griesmer} $        & $36^{\dagger}$ \\ \hline
$77 $         &           $40~\cite{08Li}$ &   -  &        $\le 38~\cite{60Griesmer} $        &     ${\bf 38^*}$     \\ \hline
$78 $         & $40~\cite{97MAGMA}$ &  -  &        $ 40~\cite{97MAGMA} $        &     -     \\ \hline
$79 $         & $\le 40~\cite{60Griesmer}$ &  ${\bf 40^*}$  &        $ 40~\cite{97MAGMA} $        &     -     \\ \hline
$80 $         &       $\le 40~\cite{08Li}$ &   ${\bf 40^*}$   &        $ 40~\cite{97MAGMA} $        &     -    \\ \hline
$81 $         &           $42~\cite{08Li}$ &   -   &        $ 40~\cite{97MAGMA} $        &     -     \\ \hline
$82 $         &           $42~\cite{08Li}$ &   -   &        $\le 40~\cite{60Griesmer} $        &     ${\bf 40^*}$     \\ \hline
$83 $         &           $44~\cite{08Li}$ &   -   &        $\le 42~\cite{60Griesmer} $        & $40^{\dagger}$ \\ \hline
$84 $         &           $44~\cite{08Li}$ &   -   &        $\le 42~\cite{60Griesmer} $        & $40^{\dagger}$ \\ \hline
$85 $         & $44~\cite{97MAGMA}$ &  -  &        $\le 42~\cite{60Griesmer} $        &     ${\bf 42^*}$     \\ \hline
$86 $         & $\le 44~\cite{60Griesmer}$ &  ${\bf 44^*}$  &        $44~\cite{97MAGMA} $        &     -     \\ \hline
$87 $         &       $\le 44~\cite{08Li}$ &  ${\bf 44^*}$  &        $44~\cite{97MAGMA} $        &     -     \\ \hline
$88 $         &           $46~\cite{08Li}$ &   -   &        $44~\cite{97MAGMA} $        &     -     \\ \hline
$89 $         &           $46~\cite{08Li}$ &   -   &        $\le 44~\cite{60Griesmer} $        &     ${\bf 44^*}$     \\ \hline
$90 $         & $48~\cite{97MAGMA}$ &  -  &        $\le 46~\cite{60Griesmer} $        & $44^{\dagger}$ \\ \hline
$91 $         &           $48~\cite{08Li}$ &   -   &        $\le 46~\cite{60Griesmer} $        & $44^{\dagger}$ \\ \hline
$92 $         &           $48~\cite{08Li}$ &   -   &        $\le 46~\cite{60Griesmer} $        &     ${\bf 46^*}$     \\ \hline
$93 $         & $48~\cite{97MAGMA}$ &  -  &        $ 48~\cite{97MAGMA} $        &     -     \\ \hline
$94 $         & $\le 48~\cite{60Griesmer}$ &  ${\bf 48^*}$  &        $ 48~\cite{97MAGMA} $        &     -     \\ \hline
$95 $         &       $\le 48~\cite{08Li}$ &  ${\bf 48^*}$  &        $ 48~\cite{97MAGMA} $        &     -     \\ \hline
$96 $         &           $50~\cite{08Li}$ &   -   &        $ 48~\cite{97MAGMA} $        &    -    \\ \hline
$97 $         &           $50~\cite{08Li}$ &   -   &        $ 48~\cite{97MAGMA} $        &    -     \\ \hline
$98 $         &           $52~\cite{08Li}$ &   -   &        $\le 48~\cite{60Griesmer} $        &     ${\bf 48^*}$     \\ \hline
$99 $         &           $52~\cite{08Li}$ &   -   &        $\le 50~\cite{60Griesmer} $        & $48^{\dagger}$ \\ \hline
$100$         & $\le 52~\cite{60Griesmer}$ &  ${\bf 52^*}$  &        $\le 50~\cite{60Griesmer} $        &     ${\bf 50^*}$
\end{tabular}
\caption{$\dso(n,k)$ for $50 < n \le 100$ and $k=4,5$}
\label{table: dso(n,k)2}
\renewcommand{\arraystretch}{1}
\end{table}

\section{Preliminaries}\label{sec: prelim}
Let $GF(q)$ be a finite field with $q$ elements. We consider the case where $q=2$ only. A subspace $\C$ of $GF(q)^n$ is called a \emph{linear code} of length $n$.  
For $n,k \in \Z^+$, a $k$-dimensional linear code $\C \subset GF(q)^{n}$ is called an \emph{$[n,k]$ code}.   The elements of $\C$ are called \emph{codewords}. A \emph{generator matrix} for $\C$ is a $k\times n$ matrix $G$ whose rows form a basis for $\C$.

For $\bx = (x_1,x_2,\ldots,x_n),\by = (y_1,y_2,\ldots,y_n) \in GF(q)^n$, the \emph{ordinary inner product} $\bx \cdot \by$ of $\bx$ and $\by$ is $\sum_{i=1}^n x_iy_i$.
For a linear code $\C$, the code
\[
\C^{\perp} := \left\{ \bx \in GF(q)^n \; \middle| \; \bx \cdot \by =0\: {\mbox{for  all }} \by \in \C\right\}.
\]
is called the \emph{dual} of $\C$. A linear code $\C$ satisfying $\C \subseteq \C^\perp$ (resp. $\C = \C^\perp$) is called \emph{self-orthogonal} (abbr. SO) (resp. \emph{self-dual}). 

The \emph{(Hamming) weight} $\wt(\bx)$ of a vector $\bx$ in $GF(q)^n$ is the total number of nonzero coordinates in $\bx$. 
For $\bx,\by \in GF(q)^n$, we define the \emph{(Hamming) distance} $d(\bx,\by)$ between $\bx$ and $\by$ by the number of coordinates in which $\bx$ and $\by$ differ. The \emph{minimum distance} of a code $\C$ is the smallest nonzero distance between any two distinct codewords. 
For $n,k,d \in \Z^+$, an \emph{$[n,k,d]$ code} $\C$ is an $[n,k]$ code whose minimum distance is $d$. 
We call an $[n,k,d]$ code $\C$ \emph{optimal} if its minimum distance $d$ is the highest among all $[n,k]$ linear codes. We denote by $d(n,k)$ the minimum distance of an optimal $[n,k]$ code.
An $[n,k,d]$ SO code $\C$ is called by \emph{optimal SO} if its minimum distance $d$ is the highest among all $[n,k]$ SO codes. We denote by $\dso(n,k)$ the minimum distance of an optimal $[n,k]$ SO code. 

For $k,d \in \Z^+$,  let $n(k,d)$ be the smallest value of $n$ for which an $[n,k,d]$ code exists. 
A notable lower bound on $n(k,d)$ was obtained by Griesmer as follows.
\begin{theorem}{\rm (\cite{60Griesmer,03HufPle})} \label{thm: Griesmer bdd}
Let $\C$ be an $[n,k,d]$ code over $GF(2)$ with $k \ge 1$. Then
\[n(k,d) \ge  g(k,d) := \sum_{i=0}^{k-1}  \ceil{\frac{d}{2^i}}.\]
\end{theorem}

Let us collect some required notations. For any $[n,k]$ code $\C$ generated by $G$, we denote by $\nmr_i(G)$ the $i$th row of $G$ from the top for $1\le i \le k$  and $\nmc_j(G)$ the $j$th column of $G$ from the left for $1 \le j \le n$. If there is no danger of confusion to the matrix $G$, then we will write $\nmr_i$ (resp. $\nmc_j$) for $\nmr_i(G)$ (resp. $\nmc_j(G)$). For a positive integer $s$, we denote by $sG = (G,G,\ldots,G)$ the juxtaposition of $s$ copies of $G$.

For $k \in \Z^+$, we denote by $\calS_k$ the $[2^k - 1, k]$ simplex code and $\nmH_k$ the generator matrix of $\calS_k$ whose $i$th column is the $k$-dimensional binary representation of $i$ for $1 \le i \le 2^k -1$. For example, the first column is written as $[0,0,\ldots,0,1]^T$.
For $i = 1,2,\ldots,2^k -1$, we let
\[
\sfh_i := \text{the $i$th column vector of $\nmH_k$ (from the left hand side)}.
\]

Let $m_1,m_2,\ldots,m_t$ be nonnegative integers and $a, n \in \Z^+$. If $n\equiv  m_1,m_2,\ldots,\text{or } m_t \pmod{a}$, then we write
\[
n\equiv_a  m_1,m_2,\ldots,m_t.
\]
Otherwise, we denote
\[
n \not\equiv_a  m_1,m_2,\ldots,m_t.
\]
If $n \equiv m_1 \pmod{a},~n \equiv m_2 \pmod{a},\ldots, \text{ and } n \equiv m_t \pmod{a}$, then we write
\[
n \equiv_a m_1  \equiv_a m_2 \cdots  \equiv_a m_t.
\] 
For a statement $P$, we define
\[
\delta(P) := \begin{cases}
1 & \text{if $P$ is true},\\
0 & \text{otherwise}.
\end{cases}
\]
For instance, for nonnegative integers $m_1,m_2,\ldots,m_t$ and $a, n \in \Z^+$, 
\[
\delta(n\equiv_a  m_1,m_2,\ldots,m_t) = \begin{cases}
1 & \text{if $n\equiv_a m_1,m_2,\ldots,m_t$},\\
0 & \text{otherwise}.
\end{cases}
\]

\section{Characterizations for self-orthogonality}\label{sec: equiv cond}
In this section, we obtain  characterizations  for self-orthogonality by reading column vectors of a generator matrix. 

For a $k \times n$ matrix $G$ and $i = 1,2,\ldots,2^k -1$, we define
\[\ell_{\sfh_i}(G) := \text{the number of $\sfh_i$ among the columns of $G$.}\]
If there is no danger of confusion to the matrix $G$, then we will write $\ell_i$ for $\ell_{\sfh_i}(G)$.

For a $k \times n$ matrix $G$ and $0 < j \le k$, we define a multiset
\begin{align*}
\bfI(j) := \left\{c_i(G) \; \middle|
\begin{array}{l}
\mathrm{(i)}~0 < i \le n,\\
\mathrm{(ii)}~\text{$c_i(G) = \sfh_{t}$ for $1\le t \le 2^k -1$}\\
\quad \text{satisfying $\floor{\frac{t}{2^{j-1}}}\equiv_2 1$}
\end{array}
\right\}.
\end{align*}
\begin{example}\label{eg: I(j)}
Let $\C_{10,3}$ be a $[10,3]$ code generated by
\[\arraycolsep=2.5pt
G_{10,3} = \left[ \begin{array}{ccccccccccc}
0&0&0&0&1&1&1&1&1&1 \\
0&1&1&1&0&0&1&1&1&1 \\
1&0&1&1&0&1&0&0&0&1 \\
\end{array}
\right].
\]
Then we have
\[
\ell_1 = \ell_2 =\ell_4 =\ell_5 = \ell_7 = 1, \quad \ell_3 = 2,\quad \text{and} \quad \ell_6 =3.
\]
and
\begin{align}\label{eq: I(1,2,3)}
\begin{aligned}
\bfI(1) &= \left\{c_1, c_3, c_4, c_6, c_{10} \right\} \\
&= \{\sfh_1, \sfh_3, \sfh_3, \sfh_5,\sfh_7\},\\
\bfI(2) &= \left\{c_2,c_3,c_4,c_7,c_8,c_9,c_{10} \right\}\\
&= \{\sfh_2, \sfh_3, \sfh_3, \sfh_6, \sfh_6, \sfh_6,\sfh_7\},\\
\bfI(3) &= \left\{ c_5, c_6, c_7, c_8, c_9, c_{10} \right\}\\
&= \{\sfh_4, \sfh_5, \sfh_6, \sfh_6, \sfh_6, \sfh_7\}.
\end{aligned}
\end{align}
\end{example}

In terms of $\bfI(j)$, we obtain the following characterization for self-orthogonality.
\begin{theorem}\label{thm: eq cond for SO}
Let $\C$ be an $[n,k]$ code generated by $G$. Then, $\C$ is SO if and only if for all $0 < j \le j' \le k$, $|\bfI(j) \cap \bfI(j')|$ is even.
\end{theorem} 
\begin{proof}
Note that for $i = 1,2,\ldots, 2^k-1$,
\[\arraycolsep=2.5pt
\sfh_i = \left[\begin{array}{c}
\delta\left(\floor{\frac{i}{2^{k-1}}} \equiv_2 1\right) \\[1ex]
\delta\left(\floor{\frac{i}{2^{k-2}}} \equiv_2 1\right) \\
\vdots \\
\delta\left(\floor{\frac{i}{2^{0}}}  \equiv_2 1\right)
\end{array} \right],
\]
that is, $\floor{\frac{i}{2^{j-1}}}$ is odd for $i = 1,2,\ldots, 2^k-1$  and $j = 1,2,\ldots,k$ if and only if the $j$th component of $\sfh_i$ from the bottom is $1$.
Therefore, $|\bfI(j) \cap \bfI(j')|$ counts the number of columns of $G$ whose $j$th and $j'$th components are both $1$.
Thus, in terms of row indices, we have
\[
|\bfI(j) \cap \bfI(j')| = \nmr_{k - j + 1} \cdot \nmr_{k - j' + 1}.
\]
for $0 < j \le j' \le k$ and hence, our assertion follows.
\end{proof}

\begin{example}\label{eg: SO example}\hfill \\
\noindent (1)  Let $\C_{10,3}$ be the $[10,3]$ code generated by $G_{10,3}$ appeared in Example~\ref{eg: I(j)}. 
From Equation~\eqref{eq: I(1,2,3)}, we see that 
\[
|\bfI(1) \cap \bfI(1)| =  |\{c_1, c_3, c_4, c_6, c_{10}\}|= 5
\]
is odd. Thus, by Theorem \ref{thm: eq cond for SO}, $\C_{10,3}$ is not SO.  \medskip

\noindent (2) Let $\tC_{10,3}$ be an $[11,3]$ code generated by
\[\arraycolsep=2.5pt
\tG_{10,3} = \left[ \begin{array}{cccccccccc||c}
0&0&0&0&1&1&1&1&1&1&0 \\
0&1&1&1&0&0&1&1&1&1&1 \\
1&0&1&1&0&1&0&0&0&1&1 \\
\end{array}
\right].
\]
It is easy to obtain $\bfI(j)$ as follows:
\begin{align*}
\bfI(1) &= \left\{ c_1,c_3,c_4,c_6,c_{10},c_{11} \right\},\\
\bfI(2) &= \left\{ c_2,c_3,c_4,c_7,c_8,c_9,c_{10},c_{11} \right\},\\
\bfI(3) &= \left\{  c_5, c_6, c_7, c_8, c_9, c_{10} \right\}.
\end{align*}
Therefore, we have
\[
\begin{array}{ll}
\bfI(1) \cap \bfI(2)= \left\{ c_3,c_4,c_{10},c_{11} \right\},\\
\bfI(1) \cap \bfI(3)= \left\{ c_6,c_{10} \right\},\\
\bfI(2) \cap \bfI(3)= \left\{ c_7,c_8,c_9,c_{10} \right\}.
\end{array}
\]
Since $|\bfI(j) \cap \bfI(j')|$ is even for all $0 < j \le j' \le 3$, by Theorem \ref{thm: eq cond for SO}, $\tC_{10,3}$ is SO. 
\end{example}

\subsection{Self-orthogonality for dimension 2 and 3}\label{subsec: SO code gen by 2,3 codewords}
In this subsection, we characterize self-orthogonality for $[n,2]$ codes in terms of $\ell_i$ using Theorem~\ref{thm: eq cond for SO}. For $[n,3]$ codes, we introduce the result in~\cite{06bou} which gives a characterization for self-orthogonality and we reprove it using Theorem~\ref{thm: eq cond for SO}.

For $[n,2]$ codes, we obtain the following characterization for self-orthogonality.

\begin{lemma}\label{lem: 2 by n SO condtion}
Let $\C$ be an $[n,2]$ code. The code $\C$ is SO if and only if $\C$ is generated by a matrix $G$ satisfying
\begin{align}\label{eq: equiv cond for SO dim 2}
\ell_1  \equiv_2 \ell_2 \equiv_2 \ell_3 \equiv_2 0.
\end{align}
\end{lemma}

\begin{proof}
Let $G$ be a generator matrix of $\C$. By Theorem~\ref{thm: eq cond for SO}, it suffices to show that Equation~\eqref{eq: equiv cond for SO dim 2} holds if and only if for all $0 < j \le j' \le 2$, $|\bfI(j) \cap \bfI(j')|$ is even. By definition of~$\bfI(j)$, we obtain that
\begin{align*}
&|\bfI(1)| = \ell_1 + \ell_3, \ \  |\bfI(2)| = \ell_2 + \ell_3, \ \  \text{and} \ \   |\bfI(1) \cap \bfI(2)| = \ell_3.
\end{align*}
This shows that for all $0 < j \le j' \le 2$, $|\bfI(j) \cap \bfI(j')|$ is even if and only if
\[
\ell_1  \equiv_2 \ell_2 \equiv_2 \ell_3 \equiv_2 0.\qedhere
\]
\end{proof}

For $[n,3]$ codes, a characterization for the self-orthogonality was introduced in~\cite{06bou}. We reprove this characterization using Theorem~\ref{thm: eq cond for SO}.

\begin{lemma}{\rm \cite[Lemma 3]{06bou}}\label{lem: 3 by n SO condtion}
Let $\C$ be an $[n,3]$ code. The code $\C$ is SO if and only if $\C$ is generated by a matrix $G$ satisfying
\begin{align}\label{eq: equiv cond for SO dim 3}
\ell_1 \equiv_2 \ell_2 \equiv_2 \cdots \equiv_2 \ell_7.
\end{align}
\end{lemma}

\begin{proof}
Let $G$ be a generator matrix of $\C$. By Theorem~\ref{thm: eq cond for SO}, it suffices to show that Equation~\eqref{eq: equiv cond for SO dim 3} holds if and only if for all $0 < j \le j' \le 3$, $|\bfI(j) \cap \bfI(j')|$ is even. By definition of~$\bfI(j)$, we obtain that
\begin{align*}
&|\bfI(1)| = \ell_1 + \ell_3 + \ell_5 + \ell_7, \quad  |\bfI(1) \cap \bfI(2)| = \ell_3 + \ell_7, \\
& |\bfI(2)| = \ell_2 + \ell_3 + \ell_6 + \ell_7, \quad |\bfI(1) \cap \bfI(3)| = \ell_5 + \ell_7,\\
&|\bfI(3)| = \ell_4+ \ell_5 + \ell_6 + \ell_7, \quad|\bfI(2) \cap \bfI(3)| = \ell_6 + \ell_7.
\end{align*}
This shows that for all $0 < j \le j' \le 3$, $|\bfI(j) \cap \bfI(j')|$ is even if and only if
\[
\ell_1 \equiv_2 \ell_2\equiv_2 \cdots \equiv_2 \ell_7. \qedhere
\]
\end{proof}

\subsection{Self-orthogonal codes of dim 4}
In this subsection, we provide a characterization for self-orthogonality for $[n,4]$ codes in terms of congruence equations on $\ell_i$'s.

Recall that for an $[n,4]$ code $\C$ generated by $G$,
\begin{align*}
\bfI(1) & = \{ c_i  \mid c_i = \sfh_t ~\text{for $t \in \{1,3,5,7,9,11,13,15\}$} \},\\
\bfI(2) & = \{ c_i  \mid c_i = \sfh_t ~\text{for $t \in \{2,3,6,7,10,11,14,15\}$} \},\\
\bfI(3) & = \{ c_i  \mid c_i = \sfh_t ~\text{for $t \in \{4,5,6,7,12,13,14,15\}$} \},\\
\bfI(4) & = \{ c_i  \mid c_i = \sfh_t ~\text{for $t \in \{8,9,10,11,12,13,14,15\}$} \}.
\end{align*}
Therefore, we have
\begin{align}
{\small\begin{aligned}\label{eq: inner product of rows dim 4}
& |\bfI(1)| = \ell_{1} + \ell_{3} + \ell_{5} + \ell_{7} + \ell_{9} + \ell_{11} + \ell_{13} + \ell_{15},\\
& |\bfI(2)| = \ell_{2} + \ell_{3} + \ell_{6} + \ell_{7} + \ell_{10} + \ell_{11} + \ell_{14} + \ell_{15}, \\
& |\bfI(3)| = \ell_{4} + \ell_{5} + \ell_{6} + \ell_{7} + \ell_{12} + \ell_{13} + \ell_{14} + \ell_{15}, \\
& |\bfI(4)| = \ell_{8} + \ell_{9} + \ell_{10} + \ell_{11} + \ell_{12} + \ell_{13} + \ell_{14} + \ell_{15},\\
& |\bfI(1) \cap \bfI(2)| = \ell_{3} + \ell_{7} + \ell_{11} + \ell_{15},\\
& |\bfI(1) \cap \bfI(3)| = \ell_{5} + \ell_{7} + \ell_{13} + \ell_{15},\\
& |\bfI(1) \cap \bfI(4)| = \ell_{9} + \ell_{11} + \ell_{13} + \ell_{15},\\
& |\bfI(2) \cap \bfI(3)| = \ell_{6} + \ell_{7} + \ell_{14} + \ell_{15},\\
& |\bfI(2) \cap \bfI(4)| = \ell_{10} + \ell_{11} + \ell_{14} + \ell_{15},\\
& |\bfI(3) \cap \bfI(4)| = \ell_{12} + \ell_{13} + \ell_{14} + \ell_{15}.
\end{aligned}}
\end{align}

We define the sets $\calP_s^{(1)}$ and $\calP_s^{(2)}$ for $s = 1,2,\ldots, 15$ as in Table~\ref{table: calP}.
Although we only use the union $\calP_s^{(1)} \cup \calP_s^{(2)}$ of these sets in the following lemma, we define them separately for later use.

\begin{lemma}\label{lem: equiv cond for SO dim 4} 
Let $\C$ be an $[n,4]$ code.
If $\C$ is SO, then any generator matrix $G$ of $\C$ satisfies that
\begin{align}\label{eq: equiv cond for SO dim 4 in lem}
\ell_{i_1} + \ell_{i_2} \equiv_2 \ell_{j_1} + \ell_{j_2}
\end{align}
for $s = 1,2,\ldots, 15$ and all $(i_1,i_2),(j_1,j_2) \in \calP_s^{(1)} \cup \calP_s^{(2)}$. 
\end{lemma}

\begin{table}[ht]
\centering
\renewcommand{\arraystretch}{1.5}
\begin{tabular}{c|c|c}
$s$ & $\calP_s^{(1)}$ & $\calP_s^{(2)}$  \\ \hline
$1$ & $\{(2,3),(4,5),(6,7)\}$ & $\{(8,9),(10,11),(12,13),(14,15) \}$  \\ \hline
$2$ & $\{(1,3),(4,6),(5,7)\}$ & $\{(8,10),(9,11),(12,14),(13,15) \}$  \\ \hline
$3$ & $\{(1,2),(4,7),(5,6)\}$ & $\{(8,11),(9,10),(12,15),(13,14) \}$  \\ \hline
$4$ & $\{(1,5),(2,6),(3,7)\}$ & $\{(8,12),(9,13),(10,14),(11,15) \}$  \\ \hline
$5$ & $\{(1,4),(2,7),(3,6)\}$ & $\{(8,13),(9,12),(10,15),(11,14) \}$  \\ \hline
$6$ & $\{(1,7),(2,4),(3,5)\}$ & $\{(8,14),(9,15),(10,12),(11,13) \}$  \\ \hline
$7$ & $\{(1,6),(2,5),(3,4)\}$ & $\{(8,15),(9,14),(10,13),(11,12) \}$  \\ \hline
$8$ & $\{(1,9),(2,10),(3,11)\}$ & $\{(4,12),(5,13),(6,14),(7,15) \}$  \\ \hline
$9$ & $\{(1,8),(2,11),(3,10)\}$ & $\{(4,13),(5,12),(6,15),(7,14) \}$  \\ \hline
$10$ & $\{(1,11),(2,8),(3,9)\}$ & $\{(4,14),(5,15),(6,12),(7,13) \}$  \\ \hline
$11$ & $\{(1,10),(2,9),(3,8)\}$ & $\{(4,15),(5,14),(6,13),(7,12) \}$  \\ \hline
$12$ & $\{(1,13),(2,14),(3,15)\}$ & $\{(4,8),(5,9),(6,10),(7,11) \}$  \\ \hline
$13$ & $\{(1,12),(2,15),(3,14)\}$ & $\{(4,9),(5,8),(6,11),(7,10) \}$  \\ \hline
$14$ & $\{(1,15),(2,12),(3,13)\}$ & $\{(4,10),(5,11),(6,8),(7,9) \}$  \\ \hline
$15$ & $\{(1,14),(2,13),(3,12)\}$ & $\{(4,11),(5,10),(6,9),(7,8) \}$ 
\end{tabular}
\caption{$\calP_s^{(t)}$ for $s=1,2,\ldots,15$ and $t = 1,2$}
\label{table: calP}
\end{table}

\begin{proof}
Since $\C$ is SO, for all $0 < j \le j' \le 4$, $|\bfI(j) \cap \bfI(j')| \equiv_2 0$ by Theorem~\ref{thm: eq cond for SO}. Thus, by Equation~\eqref{eq: inner product of rows dim 4}, we obtain the following:
\begin{itemize}
\item Since $\ell_{2} + \ell_{3} + \ell_{6} + \ell_{7} \equiv_2 |\bfI(2)| + |\bfI(2) \cap \bfI(4)| \equiv_2 0$, we have $\ell_{2} + \ell_{3} \equiv_2 \ell_{6} + \ell_{7}$.
\item Since $\ell_{4} + \ell_{5} + \ell_{6} + \ell_{7} \equiv_2 |\bfI(3)| + |\bfI(3) \cap \bfI(4)| \equiv_2 0$, we have $\ell_{4} + \ell_{5} \equiv_2 \ell_{6} + \ell_{7}$.
\item Since $\ell_6 + \ell_7 + \ell_{14} + \ell_{15} \equiv_2 |\bfI(2) \cap \bfI(3)| \equiv_2 0$, we have $\ell_6 + \ell_7 \equiv_2 \ell_{14} + \ell_{15}$.
\item Since $\ell_{10} + \ell_{11} + \ell_{14} + \ell_{15} \equiv_2 |\bfI(2) \cap \bfI(4)| \equiv_2 0$, we have $\ell_{10} + \ell_{11} \equiv_2 \ell_{14} + \ell_{15}$.
\item Since $\ell_{12} + \ell_{13} + \ell_{14} + \ell_{15} \equiv_2 |\bfI(3) \cap \bfI(4)| \equiv_2 0$, we have $\ell_{12} + \ell_{13} \equiv_2 \ell_{14} + \ell_{15}$.
\item Since $\ell_{8} + \ell_{9} + \ell_{10} + \ell_{11} \equiv_2 |\bfI(4)| + |\bfI(3) \cap \bfI(4)| \equiv_2 0$, we have $\ell_{8} + \ell_{9} \equiv_2 \ell_{10} + \ell_{11}$.
\end{itemize}
This implies the case $s = 1$ since Equation~\eqref{eq: equiv cond for SO dim 4 in lem} can be simply written as
\begin{align*}
&\ell_2 + \ell_3 \equiv_2 \ell_4 +\ell_5 \equiv_2 \ell_6 + \ell_7 \equiv_2 \ell_8 + \ell_9 \\
&\equiv_2 \ell_{10} + \ell_{11} \equiv_2 \ell_{12} + \ell_{13} \equiv_2 \ell_{14} + \ell_{15}.
\end{align*}
The rest cases can be shown in the same manner.
\end{proof}
\begin{remark}\label{rem: prop of calP}\hfill 
\begin{enumerate}[label = {\rm (\arabic*)}]
\item For $s \in \{1,2,\ldots,7\}$, 
\[
\{i \in \Z \mid \text{$i$ appears in $\calP_s^{(1)}$}\} = \{1,2,\ldots,7\} \setminus\{s\}.
\]
\item For any $i \neq j \in \{8,9,\ldots, 15\}$, there exists $s \in \{1,2,\ldots,7\}$ such that $(i,j) \in \calP_s^{(2)}$.
\end{enumerate}
\end{remark}

Now we are ready to introduce our characterization for self-orthogonality.

\begin{theorem}\label{thm: equiv cond for SO dim 4}
Let $\C$ be an $[n,4]$ code generated by $G$. The code $\C$ is SO  if and only if there is $s \in \{1,2,\ldots, 15\}$ such that for each $t = 1,2$,
\begin{align}\label{eq: cong eqn for SO dim 4}
\ell_i \equiv_2 \ell_j \quad \text{for $i,j \in \calI_s^{(t)}$}.
\end{align}
Here, $\calI_s^{(t)}$'s are sets given in Table~\ref{table: calI}.
\end{theorem}

\begin{table}[ht]
\centering
\renewcommand{\arraystretch}{1.5}
\begin{tabular}{c|c|c}
$s$ & $\calI_s^{(1)}$ & $\calI_s^{(2)}$  \\ \hline
$1$ & $\{1,2,3,4,5,6,7\}$ & $\{8,9,10,11,12,13,14,15\}$ \\ \hline
$2$ & $\{1,2,3,8,9,10,11\}$ & $\{4,5,6,7,12,13,14,15\}$ \\ \hline
$3$ & $\{1,2,3,12,13,14,15\}$ & $\{4,5,6,7,8,9,10,11\}$ \\ \hline
$4$ & $\{1,4,5,8,9,12,13\}$ & $\{2,3,6,7,10,11,14,15\}$ \\ \hline
$5$ & $\{1,4,5,10,11,14,15\}$ & $\{2,3,6,7,8,9,12,13\}$ \\ \hline
$6$ & $\{1,6,7,8,9,14,15\}$ & $\{2,3,4,5,10,11,12,13\}$ \\ \hline
$7$ & $\{1,6,7,10,11,12,13\}$ & $\{2,3,4,5,8,9,14,15\}$ \\ \hline
$8$ & $\{2,4,6,8,10,12,14\}$ & $\{1,3,5,7,9,11,13,15\}$ \\ \hline
$9$ & $\{2,4,6,9,11,13,15\}$ & $\{1,3,5,7,8,10,12,14\}$ \\ \hline
$10$ & $\{2,5,7,8,10,13,15\}$ & $\{1,3,4,6,9,11,12,14\}$ \\ \hline
$11$ & $\{2,5,7,9,11,12,14\}$ & $\{1,3,4,6,8,10,13,15\}$ \\ \hline
$12$ & $\{3,4,7,8,11,12,15\}$ & $\{1,2,5,6,9,10,13,14\}$ \\ \hline
$13$ & $\{3,4,7,9,10,13,14\}$ & $\{1,2,5,6,8,11,12,15\}$ \\ \hline
$14$ & $\{3,5,6,8,11,13,14\}$ & $\{1,2,4,7,9,10,12,15\}$ \\ \hline
$15$ & $\{3,5,6,9,10,12,15\}$ & $\{1,2,4,7,8,11,13,14\}$
\end{tabular}
\caption{$\calI_s^{(t)}$ for $s=1,2,\ldots,15$ and $t =1,2$}
\label{table: calI}
\end{table}

\begin{proof}
It will be proved in Appendix~\ref{apx: Proof of Thm 3.8}.
\end{proof}

\section{Algorithms to construct shortest SO embeddings}\label{sec: algorithm}
In this section, considering Lemmas~\ref{lem: 2 by n SO condtion},~\ref{lem: 3 by n SO condtion}, and Theorem~\ref{thm: equiv cond for SO dim 4}, we introduce an algorithm which extends an $[n,k]$ code to an SO code by adding the smallest number of columns for $k = 2,3,4$. We also introduce an algorithm which extends an $[n,5]$ code to an SO code.

\begin{definition} Let $\C$ be an $[n,k]$ code generated by $G$.
\begin{enumerate}[label = (\arabic*)]
\item An \emph{SO embedding} of $\C$ is an SO code whose generator matrix $\tG$ is obtained by adding a set $S$ of column vectors to $G$, that is,
\[
\tG := \left[ \arraycolsep=2.5pt\begin{array}{c||c} G & S  \end{array} \right].
\]
\item An SO embedding of $\C$ is called a \emph{shortest SO embedding} of $\C$ if its length is shortest among all SO embeddings of $\C$.
\end{enumerate}
\end{definition}

In Example~\ref{eg: SO example}, $\tC_{10,3}$ is a shortest SO embedding of $\C_{10,3}$.

\subsection{Algorithms for dimension 2 and 3}\label{subsec: shortest SO ext dim 2,3}

We begin with the following algorithm for two-dimensional linear codes.

\begin{algorithm}\label{alg: SO embedding dim2} \hfill
\begin{enumerate}[label=$\bullet$]
\item Input: A generator matrix $G$ of an $[n,2]$ code.
\item Output: A generator matrix $\tG$ for a shortest SO embedding.
\begin{enumerate}[label=({\bf A}\arabic*)]
\item Put $\tG \leftarrow G$ and $i \leftarrow 1$.
\item Let 
\[
\tG \leftarrow \begin{cases}
\tG & \text{if $\ell_i(\tG) \equiv_2 0$}, \\
\left[ \arraycolsep=2.5pt\begin{array}{c||c} \tG & \sfh_i  \end{array} \right] & \text{if $\ell_i(\tG) \equiv_2 1$},
\end{cases}
\]
where $\left[ \arraycolsep=2.5pt\begin{array}{c||c} \tG & \sfh_i  \end{array} \right]$ is the juxtaposition of $\tG$ and $\sfh_i$.\\[-1.5ex]
\item If $i < 3$, then put $i \leftarrow i+1$ and go to ({\bf A}2). Otherwise, terminate the algorithm. 
\end{enumerate}
\end{enumerate}
\end{algorithm}

With the resulting matrix $\tG$ of Algorithm~\ref{alg: SO embedding dim2}, we let 
\begin{align}\label{eq: def of tC dim 2}
\tC := \text{the linear code generated by $\tG$.}
\end{align}

\begin{remark}
We can obtain an SO code from an $[n,2]$ code by adding at most $3$ columns.
\end{remark}

More precisely, we have the following.

\begin{theorem}\label{thm: shortest alg dim 2}
Let $\C$ be an $[n,2]$ code generated by $G$.  Then $\tC$ is a shortest SO embedding of $\C$.
\end{theorem}
\begin{proof}
By Lemma~\ref{lem: 2 by n SO condtion}, it is obvious.
\end{proof}

\begin{example}\label{eg: counter for Kob}
Let $\C_{7,2}$ be an optimal $[7,2,4]$ code generated by 
\[
\arraycolsep=2.5pt G_{7,2} =
\begin{bmatrix}
0&0&0&1&1&1&1\\
1&1&1&0&0&0&1
\end{bmatrix}.
\]
We can construct an SO code from $\C_{7,2}$ using  Algorithm~\ref{alg: SO embedding dim2}.
In Step ({\bf A}1), we put $\tG_{7,2} \leftarrow G_{7,2}$. Note that
\[
\ell_1(\tG_{7,2}) = 3, \quad  \ell_2(\tG_{7,2}) = 3, \quad \text{and} \quad \ell_3(\tG_{7,2}) = 1.
\]
Therefore, when we apply Step ({\bf A}2) and ({\bf A}3), we put the juxtaposition of $\tG_{7,2}$, $\sfh_1$, $\sfh_2$, and $\sfh_3$ to new $\tG_{7,2}$. Thus, we obtain
\[
\tG_{7,2} =\left[\arraycolsep=2.5pt
\begin{array}{ccccccc||ccc}
0&0&0&1&1&1&1&0&1&1\\
1&1&1&0&0&0&1&1&0&1
\end{array}\right].
\]
It is easy to check that the $[10,2,6]$ code $\tC_{7,2}$ generated by $\tG_{7,2}$ satisfies the condition in Lemma \ref{lem: 2 by n SO condtion}, and thus it is a shortest  SO embedding of $\C_{7,2}$. Moreover, $\tC_{7,2}$ is optimal SO.
\end{example}

\begin{remark}\label{rem: error in 95Kob}
In \cite[Section 2]{95Kob}, Kobayashi and Takada proposed an algorithm which extends an $[n,k]$ code $\C$ to an SO code by juxtaposing $k$ column vectors to its generator matrix $G$.
In the algorithm, they first find a solution $a_0$ of the equation $r_1 \cdot r_1 + x^2 = 0$.
Next, they claim that the equation $r_1 \cdot r_i + a_0 x = 0$ has a solution for every $1 < i \le k$.
However, in case where $a_0 = 0$ and $r_1 \cdot r_i = 1$ for some $1 < i \le k$, there are no solutions to this equation.
It makes the entire method proposed in~\cite{95Kob} wrong.
Indeed, if $\C$ is an $[n,2]$ code with a generator matrix
$
\arraycolsep=1.8pt
\renewcommand{\arraystretch}{1}
G =
\begin{bmatrix}
0&1&1\\[- 1 ex]
1&0&1
\end{bmatrix}$,
then we cannot extend $\C$ to an SO code by adding any two columns to $G$.
\end{remark}

Now, let us consider the three-dimensional case.
First, for a $3 \times n$ matrix $G$ and $j = 0,1$, we set
\[
J_j(G) := \left\{i \in \{1,2,\ldots,7\} \; \middle| \; \ell_i \equiv_2 j \right\}.
\]
The basic idea of the following algorithm is that we keep adding the smallest number of columns to $G$ so that all $\ell_i$'s have the same parity.
\begin{algorithm}\label{alg: SO embedding dim3}\hfill
\begin{enumerate}[label=$\bullet$]
\item  Input: A generator matrix $G$ of an $[n,3]$ code.
\item  Output: A generator matrix $\tG$ for a shortest SO embedding.
\begin{enumerate}[label = ({\bf B}\arabic*)]
\item Put $\tG \leftarrow G$.
\item If $J_0(\tG) = \emptyset$ (i.e., $\ell_i \equiv_2 1$ for all $i = 1,2,\ldots,7$) or $J_1(\tG) = \emptyset$ (i.e., $\ell_i \equiv_2 0$ for all $i = 1,2,\ldots,7$), then terminate the algorithm. Otherwise, go to ({\bf B3}).
\item If $|J_0(\tG)| < |J_1(\tG)|$, then let $i_0$ be the smallest integer in $J_0(\tG)$. Otherwise, let $i_0$ be the smallest integer in $J_1(\tG)$.
\item Let
\[\arraycolsep=2.5pt
\tG \leftarrow \left[ \begin{array}{c||c} \tG & \sfh_{i_0} \end{array} \right].
\]
Go to ({\bf B}2).
\end{enumerate}
\end{enumerate}
\end{algorithm}

\begin{remark}
Note that when we apply Algorithm~\ref{alg: SO embedding dim3} to a generator matrix $G$ of an $[n,3]$ code, we add $|J_0(G)|$ columns to $G$ if $|J_0(G)| < |J_1(G)|$ and $|J_1(G)|$ columns to $G$ otherwise. Thus, we can obtain an SO code from an $[n,3]$ code by adding at most $3$ columns because $\min\{|J_0(G)|, |J_1(G)|\} \le \floor{7/2}$.
\end{remark}

With the resulting matrix $\tG$ of Algorithm~\ref{alg: SO embedding dim3}, let 
\begin{align*}
\tC := \text{the linear code generated by $\tG$.}
\end{align*}
Then we have the following theorem.
\begin{theorem}\label{thm: shortest alg dim 3}
Let $\C$ be an $[n,3]$ code generated by $G$.  Then $\tC$ is a shortest SO embedding of $\C$.
\end{theorem}

\begin{proof}
Note that if $\C$ is SO, then  $J_0(\tG) = \emptyset$ or $J_1(\tG) = \emptyset$ by Lemma~\ref{lem: 3 by n SO condtion}. Thus, we have $\tC = \C$. Therefore, we may assume that $\C$ is not SO.
Then we have $|J_{j}(G)| < |J_{j'}(G)|$ for some $j\ne j' \in \{0,1\}$.

From Step ({\bf B}2) to Step ({\bf B}4), one can see that Algorithm~\ref{alg: SO embedding dim3} will stop only when $J_{j}(\tG) = \emptyset$. Therefore, $\tG$ satisfies Equation~\eqref{eq: equiv cond for SO dim 3} and thus $\tC$ is SO by Lemma~\ref{lem: 3 by n SO condtion}.

To prove $\tC$ is a shortest SO embedding of $\C$, we let 
\[
\arraycolsep=2.5pt \hG := \left[\begin{array}{c||c} G & M \end{array}\right],
\]
where $M$ is a $3 \times l$ matrix for some $0< l <|J_j(G)|$ and $\hC$ be a linear code generated by $\hG$.
Since $0< l <|J_j(G)|$ and $|J_j(G)| + |J_{j'}(G)| = 7$, the inequality
\[
0 < |J_j(G)| - l \le |J_r(\hG)| \le |J_{j'}(G)| + l < 7
\]
holds for $r =0,1$.
Thus, by Lemma~~\ref{lem: 3 by n SO condtion}, $\hC$ is not SO.
\end{proof}

\begin{example} Let $\C_{10,3}$ be an optimal $[10,3,5]$ code generated by 
\[\arraycolsep=2.5pt
G_{10,3} =
\begin{bmatrix}
0&0&0&0&0&1&1&1&1&1\\
0&0&1&1&1&0&0&0&1&1\\
1&1&0&0&1&0&0&1&1&1
\end{bmatrix}.
\]
It is feasible to construct an SO code from $\C_{10,3}$ using Algorithm \ref{alg: SO embedding dim3} as follows.

In Step ({\bf B}1), put $\tG_{10,3} \leftarrow G_{10,3}$. In  Step ({\bf B}2), considering the number of the columns of $\tG_{10,3}$, we let
\[
J_0 (\tG_{10,3}) = \{1,2,4,6,7 \} \quad \text{and} \quad J_1 (\tG_{10,3}) = \{3, 5\}.
\]
Since both are nonempty sets,  proceed to Step ({\bf B}3).  In Step ({\bf B}3), let $i_0 = 3$ since $|J_1 (\tG_{10,3})| < |J_0 (\tG_{10,3})|$. In Step ({\bf B}4), put
\begin{align*}
\tG_{10,3} \leftarrow \left[\arraycolsep=2pt \begin{array}{c||c} \tG_{10,3} & \sfh_{3} \end{array} \right]=\left[\arraycolsep=2pt \begin{array}{cccccccccc||c} 
0&0&0&0&0&1&1&1&1&1&0\\
0&0&1&1&1&0&0&0&1&1&1\\
1&1&0&0&1&0&0&1&1&1&1
\end{array} \right].
\end{align*}
Next, go to Step ({\bf B}2) and repeat the steps. Now 
\[
J_0 (\tG_{10,3}) = \{1,2,3,4,6,7 \} \quad \text{and}\quad J_1 (\tG_{10,3}) = \{5\}.
\]
Since neither $J_0 (\tG_{10,3})$ nor $J_1 (\tG_{10,3})$ is an empty set, go to Step ({\bf B}3). In Step ({\bf B}3), we let $i_0 = 5$ since $|J_1 (\tG_{10,3})| < |J_0 (\tG_{10,3})|$.  Applying Step ({\bf B}4), we put
\begin{align*}
\tG_{10,3} \leftarrow \left[\arraycolsep=2pt \begin{array}{c||c} \tG_{10,3} & \sfh_{5} \end{array} \right]=\left[\arraycolsep=2pt \begin{array}{ccccccccccc||c} 
0&0&0&0&0&1&1&1&1&1&0&1\\
0&0&1&1&1&0&0&0&1&1&1&0\\
1&1&0&0&1&0&0&1&1&1&1&1
\end{array} \right].
\end{align*}
Go to Step ({\bf B}2) again. Since $J_1(\tG_{10,3}) = \emptyset$, the algorithm is terminated and we obtain the desired $\tG_{10,3}$. The code $\tC_{10,3}$ generated by $\tG_{10,3}$  is the $[12,3,6]$ SO code, which is a shortest SO embedding of $\C_{10,3}$. Moreover, $\tC_{10,3}$ is optimal SO.
\end{example}

\subsection{Algorithm for dimension 4}
In this subsection, we introduce an algorithm that embeds an $[n,4]$ code to an SO code by Theorem~\ref{thm: equiv cond for SO dim 4}. We let
\[
J_1(G) := \left\{i \in \{1,2,\ldots,15\} \; \middle| \; \ell_i \equiv_2 1  \right\}
\]
for a $4 \times n$ matrix $G$. 

The basic idea of the following algorithm is that we keep adding the smallest number of columns to $G$ so that seven and eight $\ell_i$'s each have the same parity.

\begin{algorithm}\label{alg: SO embedding dim4} \hfill
\begin{enumerate}[label=$\bullet$]
\item  Input: A generator matrix $G$ of an $[n,4]$ code. 
\item  Output: A generator matrix $\tG$ for a shortest SO embedding. 
\begin{enumerate}[label = ({\bf C}\arabic*)]
\item For $s = 1,2,\ldots, 15$, if $|\calI_s^{(1)} \cap J_1(G)| < 4$ (resp. $|\calI_s^{(1)} \cap J_1(G)| \ge 4$), then let 
\[
\fn^{(1)}_s \leftarrow |\calI_s^{(1)} \cap J_1(G)| \quad \text{(resp. $\fn^{(1)}_s \leftarrow |\calI_s^{(1)} \setminus J_1(G)|$)}.
\]
\item For $s = 1,2,\ldots, 15$, if $|\calI_s^{(2)} \cap J_1(G)| \le 4$ (resp. $|\calI_s^{(2)} \cap J_1(G)| > 4$), then let 
\[
\fn^{(2)}_s \leftarrow |\calI_s^{(2)} \cap J_1(G)| \quad \text{(resp. $\fn^{(2)}_s \leftarrow |\calI_s^{(2)} \setminus J_1(G)|$)}.
\]
\item Find the smallest $s_0 \in \{1,2,\ldots,15\}$ such that 
\[
\fn^{(1)}_{s_0} + \fn^{(2)}_{s_0} = \mathrm{min}\{\fn^{(1)}_s + \fn^{(2)}_s \mid 1 \le s \le 15\}.
\]
\item Put $\tG \leftarrow G$.
\item If $|\calI_{s_0}^{(1)} \cap J_1(\tG)| < 4$ (resp. $|\calI_{s_0}^{(1)} \cap J_1(\tG)| \ge 4$), then let  
\[
\calI^{(1)} \leftarrow \calI_{s_0}^{(1)} \cap J_1(\tG) \quad \text{(resp. $\calI^{(1)} \leftarrow \calI_{s_0}^{(1)} \setminus J_1(\tG)$).}
\]
\item  If $\calI^{(1)} \neq \emptyset$, then take the smallest $i_0 \in \calI^{(1)}$ and put
\[\arraycolsep=2.5pt
\tG \leftarrow \left[\begin{array}{c||c} \tG & \sfh_{i_0} \end{array} \right].
\]
Go to ({\bf C}5). Otherwise, go to ({\bf C}7).
\item If $|\calI_{s_0}^{(2)} \cap J_1(\tG)| \le 4$ (resp. $|\calI_{s_0}^{(2)} \cap J_1(\tG)| > 4$), then let   
\[
\calI^{(2)} \leftarrow \calI_{s_0}^{(2)} \cap J_1(\tG) \quad \text{(resp. $\calI^{(2)} \leftarrow \calI_{s_0}^{(2)} \setminus J_1(\tG)$).}
\]

\item If $\calI^{(2)} \neq \emptyset$, then take the smallest $i_0 \in \calI^{(2)}$ and put
\[\arraycolsep=2.5pt
\tG \leftarrow \left[\begin{array}{c||c} \tG & \sfh_{i_0} \end{array} \right].
\]
Go to ({\bf C}7). Otherwise, terminate the algorithm.
\end{enumerate}
\end{enumerate}
\end{algorithm}

For readers' convenience, we give Algorithm~\ref{alg: SO embedding dim4} written in MAGMA in Appendix~\ref{apx: alg dim 4}.
With the resulting matrix $\tG$ of Algorithm~\ref{alg: SO embedding dim4}, we let
\begin{align*}
\tC := \text{the linear code generated by $\tG$.}
\end{align*}
Then we have the following theorem.

\begin{theorem}\label{thm: shortest alg dim 4}
Let $\C$ be an $[n,4]$ code generated by $G$.  Then $\tC$ is a shortest SO embedding of $\C$.
\end{theorem}

\begin{proof}
It will be proved in Appendix~\ref{apx: Proof of 4.12}.
\end{proof}

\begin{remark}\label{rem: added column}
In Algorithm~\ref{alg: SO embedding dim4}, $\tG$ is obtained by juxtaposing $G$ and $\sfh_{i}$'s for $i \in \calI^{(1)}(G) \cup \calI^{(2)}(G)$, where
\begin{align*}
\calI^{(1)}(G) := \begin{cases}
\calI_{s_0}^{(1)} \cap J_1(G) & \text{if $|\calI_{s_0}^{(1)} \cap J_1(G)| < 4$},\\
\calI_{s_0}^{(1)} \setminus J_1(G) & \text{if $|\calI_{s_0}^{(1)} \cap J_1(G)| \ge 4$},
\end{cases}\\
\calI^{(2)}(G) := \begin{cases}
\calI_{s_0}^{(2)} \cap J_1(G) & \text{if $|\calI_{s_0}^{(1)} \cap J_1(G)| \le 4$},\\
\calI_{s_0}^{(2)} \setminus J_1(G) & \text{if $|\calI_{s_0}^{(1)} \cap J_1(G)| > 4$}.
\end{cases}
\end{align*} 
\end{remark}

\begin{proposition}
A shortest SO embedding of an $[n,4]$ code can be obtained by adding at most $5$ columns.
\end{proposition}
\begin{proof}
For each $s \in \{1,2,\ldots, 15\}$, denote the sets $\calI_s^{(1)}$ and $\calI_s^{(2)}$ given in Table~\ref{table: calI} by 
\begin{align*}
\{a_{s,1} < a_{s,2} < \cdots < a_{s,7}\} ~ \text{and}~ \{b_{s,1} < b_{s,2} < \cdots < b_{s,8}\},
\end{align*}
respectively.
Let
\begin{align*}
\bfv_s & = (v_{s,1}, v_{s,2}, \ldots, v_{s,15}) \\
& := (a_{s,1}, a_{s,2}, \ldots, a_{s,7}, b_{s,1}, b_{s,2}, \ldots, b_{s,8}) \in (\Z^+)^{15}.
\end{align*}

By the definition of $\bfv_s$,  one can see that $\{v_{s,1},v_{s,2}, \ldots, v_{s,15}\} = \{1,2,\ldots, 15\}$ for each $s \in \{1,2,\ldots,15\}$. Therefore, for each $s \in \{1,2,\ldots,15\}$, there is a permutation $\sigma_s$ of $\{1,2,\ldots,15\}$  such that
\begin{align*}
\sigma_s(\bfv_s) & := \left(\sigma_s(v_{s,1}), \sigma_s(v_{s,2}), \ldots, \sigma_s(v_{s,15})\right)\\
& = (1,2,\ldots, 15) = \bfv_1.
\end{align*}
It is checked by MAGMA that for each $s, i \in \{1,2, \ldots, 15\}$, there exists $j(s,i) \in \{1,2, \ldots, 15\}$ such that
\begin{align}\label{eq: sigma_s I}
\begin{aligned}
\sigma_s (\calI_{i}^{(1)}) &= \left\{\sigma_s(a_{i,1}), \ldots, \sigma_s(a_{i,7}) \right\} = \calI_{j(s,i)}^{(1)},\\
\sigma_s (\calI_{i}^{(2)}) &= \left\{\sigma_s(b_{i,1}), \ldots, \sigma_s(b_{i,8}) \right\} = \calI_{j(s,i)}^{(2)}.
\end{aligned}
\end{align}

Assume that Algorithm~\ref{alg: SO embedding dim4} is applied to a generator matrix $G$ of an $[n,4]$ code. Note that $s_0 \in \{1,2,\ldots, 15\}$ is obtained in Step ({\bf C}3).
For each $i \in \{1,2,\ldots, 15\}$, Equation~\eqref{eq: sigma_s I} gives that
\begin{align}\label{eq: sigma inv card}
\begin{aligned}
\calI_i^{(t)} \cap J_1(G) & = \sigma_{s_0}^{-1}\left( \calI_{j(s_0,i)}^{(t)}\right) \cap J_1(G),\\
\calI_i^{(t)} \setminus J_1(G) & = \sigma_{s_0}^{-1}\left( \calI_{j(s_0,i)}^{(t)} \right) \setminus J_1(G)
\end{aligned}
\end{align}
for $t = 1,2$.
Note that $\sigma_{s_0}$ is a bijection. Then, by Equation~\eqref{eq: sigma inv card}, for $t = 1,2$,
\begin{align}\label{eq: sigma_s same cardinal}
\begin{aligned}
|\calI_i^{(t)} \cap J_1(G)| & = |\calI_{j(s_0,i)}^{(t)} \cap \sigma_{s_0}(J_1(G))|,\\
|\calI_i^{(t)} \setminus J_1(G)| & = |\calI_{j(s_0,i)}^{(t)} \setminus \sigma_{s_0}(J_1(G))|.
\end{aligned}
\end{align}

Let us consider a new algorithm obtained by replacing
\begin{align}\label{eq: replacing}
\begin{aligned}
&\calI_i^{(t)} \cap J_1(G)~\text{by}~\calI_{i}^{(t)} \cap \sigma_{s_0}(J_1(G)),\\
&\calI_i^{(t)} \setminus J_1(G)~\text{by}~\calI_{i}^{(t)} \setminus \sigma_{s_0}(J_1(G)),~\text{and}\\
&\sfh_{i_0}~\text{by}~\sfh_{\sigma_{s_0}^{-1}(i_0)}
\end{aligned}
\end{align}
in Algorithm~\ref{alg: SO embedding dim4}. To emphasize the difference between this new algorithm and Algorithm~\ref{alg: SO embedding dim4}, we denote the integer $\fn_s^{(1)}$(resp. $\fn_s^{(2)}$) obtained in Step ({\bf C}1) (resp. ({\bf C}2)) of the new algorithm by $\hfn_s^{(1)}$(resp. $\hfn_s^{(2)}$), $s_0$ in Step ({\bf C}3) by $\hs_0$, and $\tG$ by $\hG$. Then, by Equation~\eqref{eq: sigma_s same cardinal}, we have
\[
\fn_s^{(1)} = \hfn_{j(s_0,s)}^{(1)} \quad \text{and} \quad \hfn_s^{(2)} = \fn_{j(s_0,s)}^{(2)}.
\]
Therefore, the facts that $j(s_0,s_0) = 1$ and
\[
\fn^{(1)}_{s_0} + \fn^{(2)}_{s_0} = \mathrm{min}\{\fn^{(1)}_s + \fn^{(2)}_s \mid 1 \le s \le 15\}
\]
imply $\hs_0 = 1$.

Note that Remark~\ref{rem: added column} can be modified as follows: $\hG$ is obtained by juxtaposing $G$ and $\sfh_{\sigma_{s_0}^{-1}(i)}$'s for $i \in \wcI^{(1)}(G) \cup \wcI^{(2)}(G)$, where
\begin{align*}
\wcI^{(1)}(G) := \begin{cases}
\calI_{1}^{(1)} \cap \sigma_{s_0}(J_1(G)) & \text{if $|\calI_{1}^{(1)} \cap \sigma_{s_0}(J_1(G))| < 4$},\\
\calI_{1}^{(1)} \setminus \sigma_{s_0}(J_1(G)) & \text{if $|\calI_{1}^{(1)} \cap \sigma_{s_0}(J_1(G))| \ge 4$},
\end{cases}\\
\wcI^{(2)}(G) := \begin{cases}
\calI_{1}^{(2)} \cap \sigma_{s_0}(J_1(G)) & \text{if $|\calI_{1}^{(2)} \cap \sigma_{s_0}(J_1(G))| \le 4$},\\
\calI_{1}^{(2)} \setminus \sigma_{s_0}(J_1(G)) & \text{if $|\calI_{1}^{(2)} \cap \sigma_{s_0}(J_1(G))| > 4$}.
\end{cases}
\end{align*}

One can see that
\begin{align}\label{eq: I and hI equal}
\{\sfh_i \mid i \in \calI^{(t)}(G)\} = \{\sfh_{\sigma_{s_0}^{-1}(i)} \mid i \in \wcI^{(t)}(G)\}.
\end{align}
for $t = 1,2$. 
For instance, assume that $|\calI_{s_0}^{(1)} \cap J_1(G)| < 4$. Then, by Equation~\eqref{eq: sigma_s same cardinal}, $|\calI_{j(s_0,i)}^{(t)} \cap \sigma_{s_0}(J_1(G))| < 4$. Thus, we have
\begin{align*}
\calI^{(1)}(G) = \calI_{s_0}^{(1)} \cap J_1(G) ~\; \text{and} ~\; \wcI^{(1)}(G) = \calI_{1}^{(1)} \cap \sigma_{s_0}(J_1(G)).
\end{align*}
Observe that
\begin{align*}
\{\sfh_i \mid i \in \calI^{(1)}(G) \} & = \{\sfh_i \mid i \in \calI_{s_0}^{(1)} \cap J_1(G) \} \\
& = \left\{\sfh_{\sigma_{s_0}^{-1}(i)} \; \middle|\; i \in \sigma_{s_0}\left(\calI_{s_0}^{(1)} \cap J_1(G)\right) \right\} \\
& = \left\{\sfh_{\sigma_{s_0}^{-1}(i)} \; \middle|\; i \in \calI_{1}^{(1)} \cap \sigma_{s_0}(J_1(G)) \right\} \\
& = \left\{\sfh_{\sigma_{s_0}^{-1}(i)} \; \middle|\; i \in \wcI^{(1)}(G) \right\}.
\end{align*}

By Remark~\ref{rem: added column}, the above modified remark, and Equation~\eqref{eq: I and hI equal}, $\tG$ and $\hG$ are equal up to column permutations. This implies that it suffices to consider all the possible cases only when $s=1$ to show the extended length by Algorithm \ref{alg: SO embedding dim4} is at most five. 

For the remaining part of the proof, we prove the following two claims.\smallskip

\noindent{\bf Claim 1.}
The difference between the length of $\tC$ and that of $\C$  cannot  be $7$.\smallskip

\noindent{\it Proof.\;\;}It is easy to know that $0 \le \fn^{(1)}_1 + \fn^{(2)}_1 \le 7$ in Step $({\bf C} {3})$ of Algorithm \ref{alg: SO embedding dim4} since $0\le \fn^{(1)}_1 \le 3$ and $0\le \fn^{(2)}_1 \le 4$. Note that $\fn^{(1)}_1 + \fn^{(2)}_1 =7$ only when $\fn^{(1)}_1 =3$ and $\fn^{(2)}_1 = 4$. Therefore, the number of all possible cases is $\binom {7}{3}\cdot\binom{8}{4}=2450$. The values of $\mathrm{min}\{\fn^{(1)}_s + \fn^{(2)}_s \mid 1 \le s \le 15\}$ for  such cases in Step $({\bf C} {3})$ are checked by MAGMA to find out that they are at most five.\smallskip

\noindent{\bf Claim 2.}
The difference between the length of $\tC$ and that of $\C$  cannot be $6$.\smallskip

\noindent{\it Proof.\;\;} Similarly, $\fn^{(1)}_1 + \fn^{(2)}_1 =6$ only when $\fn^{(1)}_1 = \fn^{(2)}_1 = 3$,  or $\fn^{(1)}_1 =2$ and $\fn^{(2)}_1 = 4$. Therefore, the number of possible cases are $\binom {7}{3}\cdot \binom{8}{3} + \binom {7}{2}\cdot\binom{8}{4}= 1960 + 1470=3430$. The values of $\mathrm{min}\{\fn^{(1)}_s + \fn^{(2)}_s \mid 1 \le s \le 15\}$ for  such cases in Step $({\bf C} {3})$ are checked by MAGMA to find out that they are at most five.\smallskip

Hence a shortest SO embedding of an $[n,4]$ code can be constructed by juxtaposing at most five columns.
\end{proof}

\begin{example}\label{example:dim4} \hfill\\
\noindent  (1) Let $\C_{7,4}$ be the \emph{$[7,4]$ Hamming  code} and take its generator matrix
\begin{align*}
G_{7,4} = \left[\arraycolsep=2.5pt \begin{array}{ccccccc} 
0&0&0&0&1&1&1\\
0&0&1&1&0&0&1\\
0&1&0&1&0&1&0\\
1&0&0&1&0&1&1
\end{array} \right].
\end{align*}

Since $\ell_1 = \ell_2 = \ell_4 = \ell_7 = \ell_8 = \ell_{11} = \ell_{13}= 1$, we have
\begin{align*}
J_1(G_{7,4}) &= \left\{i \in \{1,2,\ldots,15\} \; \middle| \; \ell_i \equiv_2 1  \right\}\\
& = \{1,2,4,7,8,11,13\}.
\end{align*}
Applying Steps ({\bf C}1) and ({\bf C}2), we have
\[\arraycolsep=3pt
\begin{array}{c|c|c|c|c|c|c|c|c|c|c|c|c|c|c|cl}
s           &1&2&3&4&5&6&7&8&9&10&11&12&13&14&15&  \\ \cline{1-16}
\fn_s^{(1)} &3&3&3&3&3&3&3&3&3& 3& 3& 3& 3& 3& 0 &\\ \cline{1-16}
\fn_s^{(2)} &3&3&4&4&4&3&3&4&3& 3& 4& 4& 4& 4& 1&.
\end{array}
\]
Thus, $s_0$ is set to be $15$ in Step ({\bf C}3). Recall that
\[
\calI_{15}^{(1)} = \{3,5,6,9,10,12,15\}
\]
and
\[
\calI_{15}^{(2)} = \{1,2,4,7,8,11,13,14\}. 
\]
In Step ({\bf C}4), we put $\tG_{7,4} \leftarrow G_{7,4}$ and in Step ({\bf C}5), we let
\[
\calI^{(1)} \leftarrow \calI_{15}^{(1)} \cap J_1(\tG_{7,4}) = \emptyset.
\]
Thus, we pass Step ({\bf C}6) and go to Step ({\bf C}7). In Step ({\bf C}7), let
\[
\calI^{(2)} \leftarrow \calI_{15}^{(2)} \setminus J_1(\tG_{7,4}) = \{14\}.
\]
Applying Step ({\bf C}8), we put
\begin{align*}
\tG_{7,4} \leftarrow \left[\arraycolsep=2.5pt \begin{array}{ccccccc||c} 
0&0&0&0&1&1&1&1\\
0&0&1&1&0&0&1&1\\
0&1&0&1&0&1&0&1\\
1&0&0&1&0&1&1&0
\end{array} \right].
\end{align*}
Go to Step ({\bf C}7) and repeat the steps. Since $J_1(\tG_{7,4}) = \emptyset$, 
\[
\calI^{(2)} \leftarrow \calI_{15}^{(2)} \cap J_1(\tG_{7,4}) = \emptyset
\]
and thus, the algorithm is terminated.

The code $\tC_{7,4}$ generated by $\tG_{7,4}$ is a shortest SO embedding of $\C_{7,4}$, which is optimal as well.  Notice that $\tC_{7,4}$ is the \emph{$[8,4]$ extended Hamming  code}. \medskip

\noindent (2) Let $\C_{5,4}$ be an optimal $[5,4,2]$ code generated by
\[\arraycolsep=2.5pt
G_{5,4} = \left[
\begin{array}{ccccc}
0&0&0&1&1\\
0&0&1&0&1\\
0&1&0&0&1\\
1&0&0&0&1
\end{array}
\right].
\]
Note that $\ell_1 = \ell_2 = \ell_4 = \ell_8 = \ell_{15}$. Therefore,
\begin{align*}
J_1(G_{5,4}) & = \left\{i \in \{1,2,\ldots,15\} \; \middle| \; \ell_i \equiv_2 1  \right\} \\
& = \{1,2,4,8,15\}.
\end{align*}
Applying Step ({\bf C}1) and Step ({\bf C}2), we have
\[\arraycolsep=3pt
\begin{array}{c|c|c|c|c|c|c|c|c|c|c|c|c|c|c|cl}
s           &1&2&3&4&5&6&7&8&9&10&11&12&13&14&15 & \\ \cline{1-16}
\fn_s^{(1)} &3&3&3&3&3&3&1&3&3& 3& 1& 3& 1& 1& 1 &\\ \cline{1-16}
\fn_s^{(2)} &2&2&2&2&2&2&4&2&2& 2& 4& 2& 4& 4& 4 &.
\end{array}
\]
Thus, $s_0$ is set to be $1$ in Step ({\bf C}3). Recall that
\[
\calI_{1}^{(1)} = \{1,2,3,4,5,6,7\}
\]
and
\[
\calI_{1}^{(2)} = \{8,9,10,11,12,13,14,15\}. 
\]
In Step ({\bf C}4), we put $\tG_{5,4} \leftarrow G_{5,4}$ and in Step ({\bf C}5), we let
\[
\calI^{(1)} \leftarrow \calI_{1}^{(1)} \cap J_1(\tG_{5,4}) = \{1,2,4\} .
\]
Applying Step ({\bf C}6), we put
\[\arraycolsep=2.5pt
\tG_{5,4} \leftarrow \left[
\begin{array}{ccccc||c}
0&0&0&1&1&0\\
0&0&1&0&1&0\\
0&1&0&0&1&0\\
1&0&0&0&1&1
\end{array}
\right].
\]
Go to Step ({\bf C}5) and repeat the steps. Then, we have
\[\arraycolsep=2.5pt
\tG_{5,4} \leftarrow \left[
\begin{array}{ccccc||ccc}
0&0&0&1&1&0&0&0\\
0&0&1&0&1&0&0&1\\
0&1&0&0&1&0&1&0\\
1&0&0&0&1&1&0&0
\end{array}
\right].
\]
and $J_1(\tG_{5,4}) = \{8,15\}$.
Therefore, we have 
\[
\calI^{(1)} = \calI_{1}^{(1)} \cap J_1(\tG_{5,4}) = \emptyset
\]
and thus, go to Step ({\bf C}7). In Step ({\bf C}7), let
\[
\calI^{(2)} \leftarrow \calI_{1}^{(2)} \cap J_1(\tG_{5,4}) = \{8,15\}.
\]
Applying Step ({\bf C}8), we put
\[\arraycolsep=2.5pt
\tG_{5,4} \leftarrow \left[
\begin{array}{cccccccc||c}
0&0&0&1&1&0&0&0&1\\
0&0&1&0&1&0&0&1&0\\
0&1&0&0&1&0&1&0&0\\
1&0&0&0&1&1&0&0&0
\end{array}
\right].
\]
Go to Step ({\bf C}7) and repeat the steps. Then we have
\[\arraycolsep=2.5pt
\tG_{5,4} \leftarrow \left[
\begin{array}{cccccccc||cc}
0&0&0&1&1&0&0&0&1&1\\
0&0&1&0&1&0&0&1&0&1\\
0&1&0&0&1&0&1&0&0&1\\
1&0&0&0&1&1&0&0&0&1
\end{array}
\right]
\]
and $J_1(\tG_{5,4}) = \emptyset$.
Therefore,
\[
\calI^{(2)} \leftarrow \calI_{1}^{(2)} \cap J_1(\tG_{5,4}) = \emptyset
\]
and thus, we terminate the algorithm in Step ({\bf C}8).

The code $\tC_{5,4}$ generated by $\tG_{5,4}$ is an $[10,4,4]$ SO code, which is a shortest SO embedding of $\C_{5,4}$. Moreover, $\tC_{5,4}$ is optimal SO.
\end{example}

\begin{remark}\label{rem: Alg dim 4 modify}\hfill
\begin{enumerate}[label = {\rm (\arabic*)}]
\item For simplicity of the algorithm, we chose $s_0$ to be the smallest among $1\le s\le 15$ satisfying the condition 
\[
\fn^{(1)}_{s_0} + \fn^{(2)}_{s_0} = \mathrm{min}\{\fn^{(1)}_s + \fn^{(2)}_s \mid 1 \le s \le 15\}
\]
in Step  ({\bf C}3) of Algorithm \ref{alg: SO embedding dim4}. Notice that even if we choose any element $1\le s_0' \le 15$ satisfying
\[
\fn^{(1)}_{s_0'} + \fn^{(2)}_{s_0'} = \mathrm{min}\{\fn^{(1)}_s + \fn^{(2)}_s \mid 1 \le s \le 15\},
\]
Algorithm~\ref{alg: SO embedding dim4} still gives a shortest SO embedding.
\item For simplicity of the algorithm, when $|\calI_{s_0}^{(2)} \cap J_1(\tG)| = 4$, we let
\[
\calI^{(2)} \leftarrow \calI_{s_0}^{(2)} \cap J_1(\tG)
\]
in Step ({\bf C}7) of Algorithm \ref{alg: SO embedding dim4}. However, even if we let
\[
\calI^{(2)} \leftarrow \calI_{s_0}^{(2)} \setminus J_1(\tG),
\]
Algorithm~\ref{alg: SO embedding dim4} still gives a shortest SO embedding.
\end{enumerate}
\end{remark}

\subsection{Algorithm for dimension 5}
In this subsection,  we introduce an algorithm that embeds an $[n,5]$ code to an SO code using Algorithm~\ref{alg: SO embedding dim4}. Recall that for a matrix $G$, we denote by $r_i(G)$ the $i$th row of $G$. 

The basic idea of the following algorithm is as follows.
\begin{enumerate}[label = -]
\item We add at most two columns to $G$ to obtain a matrix $G_1$ satisfying $r_{j_0}(G_1) \cdot r_i(G_1) \equiv_2 0$ for $j_0 \in \{1,2,\ldots,5\}$ and all $1 \le i \le 5$.
\item For the submatrix $G_2$ which consists of the rows of $G_1$ except $r_{j_0}(G_1)$, apply Algorithm~\ref{alg: SO embedding dim4} (let $\tG_2$ be the resulting matrix).
\item Juxtapose $r_{j_0}(G_1)$ and a zero vector horizontally so that the resulting vector $\widetilde{r}_{j_0}(G_1)$ has the same number of coordinates to the number of columns of $\tG_2$.
\item Attach $\widetilde{r}_{j_0}(G_1)$ and $\tG_2$ vertically to obtain the matrix  $\tG$.
\end{enumerate}

\begin{algorithm}\label{alg: SO embedding dim5}\hfill
\begin{enumerate}[label=$\bullet$]
\item  Input: A generator matrix $G$ of an $[n,5]$ code. 
\item  Output: A generator matrix $\tG$ for an SO embedding. 
\begin{enumerate}[label = ({\bf D}\arabic*)]
\item If there is $j_0 \in \{1,2,3,4,5\}$ such that $r_{j_0}(G) \cdot r_i(G) \equiv_2 0$ for all $1 \le i \le 5$, then put $G_1 \leftarrow G$, let $G_2$ be the matrix obtained by deleting the $j_0$th row of $G_1$, and go to Step ({\bf D}3). Otherwise, go to Step ({\bf D}2).
\item Do Step ({\bf D}2-1) or Step ({\bf D}2-2).
\begin{enumerate}[label = ({\bf D}2-\arabic*)]
\item Suppose $r_{j_0}(G) \cdot r_{j_0}(G) \equiv_2 1$ for some $1 \le j_0 \le 5$. Then put
\[G_1 \leftarrow \left[ \arraycolsep=2.5pt\begin{array}{c||c} 
r_1(G) & y_1\\
\vdots & \vdots \\
r_{j_0}(G) & 1\\
\vdots & \vdots \\
r_5(G) & y_5 
\end{array} \right]. \]
where $y_i = r_i(G) \cdot r_{j_0}(G)$. Let $G_2$ be the matrix obtained by deleting the $j_0$th row of $G_1$. Go to  Step ({\bf D}3).
\item Otherwise, that is, if $r_{j}(G) \cdot r_{j}(G) \equiv_2 0$ for all $1 \le j \le 5$, then let $j_0 = 1$ so that $r_{j_0}(G) = r_1(G)$ and put
\[
G_1 \leftarrow \left[ \arraycolsep=2.5pt\begin{array}{c||cc} 
r_1(G) & 1		& 1 \\
r_2(G) & y_2 	& 0\\
\vdots & \vdots	& \vdots\\
r_5(G) & y_5 	& 0
\end{array} \right].
\]
where $y_i = r_i(G) \cdot r_{1}(G)$. Let $G_2$ be the matrix obtained by deleting the first row of $G_1$. Go to Step ({\bf D}3).
\end{enumerate}
\item Let $\tG_2$ be the resulting matrix when we input $G_2$ into Algorithm~\ref{alg: SO embedding dim4} and let $l$ be the difference between the number of columns of $\tG_2$ and $G_2$. Go to Step ({\bf D}4).
\item Let $\widetilde{r}_{j_0}(G_1)$ be the vector obtained by juxtaposing $r_{j_0}(G_1)$ and the zero vector $\mathbf{0}$ of length $l$. Go to Step ({\bf D}5)
\item Let $\tG$ be the matrix obtained by putting together $\tG_2$ and $\widetilde{r}_{j_0}(G_1)$. Terminate the algorithm.
\end{enumerate}
\end{enumerate}
\end{algorithm}

\begin{remark}
We can obtain an SO code from an $[n,5]$ code by adding at most $2$ columns using algorithm~\ref{alg: SO embedding dim5} and adding at most $5$ columns using algorithm~\ref{alg: SO embedding dim4}. Therefore, we need at most $7$ columns to get an SO code from an $[n,5]$ code.
\end{remark}

With the resulting matrix $\tG$ of Algorithm~\ref{alg: SO embedding dim5}, we let
\begin{align*}
\tC := \text{the linear code generated by $\tG$.}
\end{align*}
Then we have the following theorem.

\begin{theorem}
Let $\C$ be an $[n,5]$ code generated by $G$. Then $\tC$ is an SO embedding of $\C$.
\end{theorem}
\begin{proof}
By following the procedure of Algorithm~\ref{alg: SO embedding dim5} we easily see that the code $\tC$ is an SO embedding.
\end{proof}

\begin{example}
Let $\C_{9,5}$ be an optimal $[9,5,3]$ code generated by
\[\arraycolsep=2.5pt
G_{9,5} = \left[
\begin{array}{ccccccccc}
0&0&0&0&1&1&1&1&1\\
0&0&0&1&0&0&1&1&1\\
0&0&1&0&0&1&0&1&1\\
0&1&0&0&0&1&1&0&1\\
1&0&0&0&0&1&1&1&0
\end{array}
\right].
\]

Since there are no $i_0 \in \{1,2,3,4,5\}$ such that $r_{i_0}(G) \cdot r_i(G) \equiv_2 0$ for all $1\le i \le 5$, we pass Step ({\bf D}1). In Step ({\bf D}2), since $ r_1(G)  \cdot  r_1(G) \equiv_2 1$, we apply Step ({\bf D}2-1). In Step ({\bf D}2-1), we put
\[\arraycolsep=2.5pt
G_1 \leftarrow \left[
\begin{array}{c||c}
\multirow{5}{*}{$G$} & 1\\
& 1 \\ 
& 1 \\
& 1 \\
& 1 \\
\end{array}\right]
=
\left[
\begin{array}{ccccccccc||c}
0&0&0&0&1&1&1&1&1&1\\
0&0&0&1&0&0&1&1&1&1\\
0&0&1&0&0&1&0&1&1&1\\
0&1&0&0&0&1&1&0&1&1\\
1&0&0&0&0&1&1&1&0&1
\end{array}
\right]
\]
and by deleting $r_1(G)$, let
\[\arraycolsep=2.5pt
G_2 =
\left[
\begin{array}{ccccccccc||c}
0&0&0&1&0&0&1&1&1&1\\
0&0&1&0&0&1&0&1&1&1\\
0&1&0&0&0&1&1&0&1&1\\
1&0&0&0&0&1&1&1&0&1
\end{array}
\right].
\]
In Step ({\bf D}3), we let
\[\arraycolsep=2.5pt
\tG_2 = 
\left[
\begin{array}{cccccccccc||c}
0&0&0&1&0&0&1&1&1&1&1\\
0&0&1&0&0&1&0&1&1&1&1\\
0&1&0&0&0&1&1&0&1&1&1\\
1&0&0&0&0&1&1&1&0&1&1
\end{array}
\right]
\]
and $l = 1$. In Step ({\bf D}4), we let 
\[
\arraycolsep=2.5pt \widetilde{r}_1(G_1) = \left[\begin{array}{cccccccccc||c}
0&0&0&0&1&1&1&1&1&1 &0 \end{array}\right].
\]
In Step ({\bf D}5), let
\[\arraycolsep=2.5pt
\tG_{9,5} = \left[
\begin{array}{c}
\widetilde{r}_1(G_1) \\ 
\tG_2
\end{array}\right]
=
\left[
\begin{array}{ccccccccccc}
0&0&0&0&1&1&1&1&1&1&0 \\ 
0&0&0&1&0&0&1&1&1&1&1 \\
0&0&1&0&0&1&0&1&1&1&1 \\
0&1&0&0&0&1&1&0&1&1&1 \\
1&0&0&0&0&1&1&1&0&1&1
\end{array}
\right].
\]

The code $\tC_{9,5}$ generated by $\tG_{9,5}$ is an $[11,5,4]$ SO code, which is an SO embedding of $\C_{9,5}$. Moreover, $\tG_{9,5}$ is optimal SO.
\end{example}

\begin{remark}
One can embed an $[n,6]$ code to an SO code in the same manner as Algorithm~\ref{alg: SO embedding dim5}.
In fact, Algorithm~\ref{alg: SO embedding dim5} can be generalized recursively to construct an SO embedding for higher-dimensional linear codes. We leave this generalized SO embedding algorithm written in MAGMA in Appendix~\ref{apx: alg dim ge 5}.
\end{remark}

\section{The minimum distances of optimal linear codes and optimal SO codes}\label{sec: min dist}

In this section, for $k = 1,2,3,4,5$, we enumerate an upper bound on $d(n,k)$ given by Griesmer~\cite{60Griesmer}
and the explicit formula for $d(n,k)$ given by Baumert and McEliece\cite{73BauMc}.
Furthermore, we calculate explicit formulas of $\dso(n,k)$ when $k=1,2,3,4$, and  $\dso(n,5)$ when $n\not\equiv_{31} 6,13,14,21,22,28,29$.

First, let us consider the one-dimensional case. It is clear that for each $n \in \Z^+$, the only optimal $[n,1]$ code is the repetition code.
On the other hand, it is also trivial  that an optimal $[n,1]$ SO code for an even integer $n$ is the $[n, 1, n]$ repetition code and that for an odd $n$ is the $[n,1,n-1]$ code which is obtained from $[n-1,1,n-1]$ repetition code by attaching a zero coordinate. Thus, we have
\[
d(n,1) = n \quad \text{and}\quad \dso(n,1) = \begin{cases}
n & \text{if $n \equiv_2 0$},\\
n-1 & \text{if $n \equiv_2 1$}.
\end{cases}
\]

Next, let us consider the two-dimensional case. 
The optimal $[n,2]$ codes are introduced in~\cite{00Jaf}. 
On the other hand, in \cite[Section 3]{06bou}, the authors explained that optimal $[n,2]$ codes are SO only for $n = 6m, 6m+1, 6m+4$ and constructed all optimal $[n,2]$ SO codes. 
Summing it up, one can immediately obtain the following theorem.

\begin{theorem}{\rm (\cite{06bou, 73BauMc, 00Jaf})}
\label{thm: optimum distance dim 2} 
The following explicit formulas hold.
\begin{enumerate}[label = {\rm (\arabic*)}]
\item For $n \ge 2$, we have
\begin{align*}
d(n,2) = \floor{2n/3}.
\end{align*}
\item For $n \ge 4$,
\begin{align*}
\dso(n,2) = \begin{cases}
\floor{2n/3} & \text{if $n \equiv_6 0,1,4$,} \\
\floor{2n/3} -1 & \text{if $n \equiv_6 2,5$,}\\
\floor{2n/3} -2 & \text{if $n \equiv_6 3$.}
\end{cases}
\end{align*}
\end{enumerate}
\end{theorem}

Theorem~\ref{thm: optimum distance dim 2} implies the following corollary.
\begin{corollary}
For $n \ge 4$, $d(n,2) = \dso(n,2)$ if and only if $n \equiv_6 0,1,4$.
\end{corollary}


Hereafter, we will also leave our proof for known results for the reader's understanding.

Now, let us consider the three-dimensional case. The Griesmer bound says that all $[n,3,d]$ codes satisfy that 
\begin{align}\label{eq: Greismer bound for dim 3}
n \ge d + \ceil{\frac{d}{2}} + \ceil{\frac{d}{4}} = \begin{cases}
7d/4& \text{if $d \equiv_4 0$},\\
(7d+5)/4& \text{if $d \equiv_4 1$},\\
(7d+2)/4& \text{if $d \equiv_4 2$},\\
(7d+3)/4& \text{if $d \equiv_4 3$}.
\end{cases}
\end{align}
Manipulating this bound, one can obtain the following upper bound on $d(n,3)$.
\begin{lemma}{\rm (\cite{60Griesmer})}
\label{lem: d(n,k) upper bound for dim 3}
For $n \ge 3$, 
\begin{align*}
d(n,3) & \le \floor{4n/7} - \delta(n \equiv_7 2) \\
& = \begin{cases}
\floor{4n/7} & \text{if $n \not\equiv_7 2$},\\
\floor{4n/7} - 1 & \text{if $n \equiv_7 2$}. 
\end{cases}
\end{align*}
\end{lemma}

\begin{proof}
From Equation~\eqref{eq: Greismer bound for dim 3}, we have
\begin{align*}
d(n,3) &\le \max\left(\hspace{-1ex}
\begin{array}{l}
\left\{ d_1 \in \Z_{\ge 0} \; \middle| \; d_1 \equiv_4 0 \   \text{and} \   \dfrac{7d_1}{4}\le n \right\}\\[2ex]
\cup \left\{ d_2 \in \Z_{\ge 0} \; \middle| \; d_2 \equiv_4 1 \   \text{and} \   \dfrac{7d_2 + 5}{4} \le n \right\}\\[2ex]
\cup \left\{ d_3 \in \Z_{\ge 0} \; \middle| \; d_3 \equiv_4 2 \   \text{and} \   \dfrac{7d_3 + 2}{4} \le n \right\}\\[2ex]
\cup \left\{ d_4 \in \Z_{\ge 0} \; \middle| \; d_4 \equiv_4 3 \  \text{and} \  \dfrac{7d_4 + 3}{4} \le n \right\}
\end{array}
\hspace{-1ex}\right)\\[3ex]
&= \max\left(\hspace{-1ex}
\begin{array}{l}
\left\{ \dfrac{4t_1}{7} \; \middle| \; t_1 \equiv_7 0 \   \text{and} \   0 \le t_1 \le n \right\}\\[2ex]
\cup \left\{ \dfrac{4t_2 - 5}{7} \; \middle| \; t_2 \equiv_7 3 \   \text{and} \   0 \le t_2 \le n \right\}\\[2ex]
\cup \left\{ \dfrac{4t_3 - 2}{7} \; \middle| \; t_3 \equiv_7 4 \   \text{and} \   0 \le t_3 \le n \right\}\\[2ex]
\cup \left\{ \dfrac{4t_4 - 3}{7} \; \middle| \; t_4 \equiv_7 6 \   \text{and} \   0 \le t_4 \le n \right\}
\end{array}
\hspace{-1ex}\right),
\end{align*}
by substituting
\[
t_1 = \dfrac{7d_1}{4},~ t_2 = \dfrac{7d_2 + 5}{4},~ t_3 = \dfrac{7d_3 + 2}{4},~ \text{and}~ t_4 = \dfrac{7d_4 + 3}{4}.
\]
This gives us that
\begin{align*}
d(n,3) & \le \begin{cases}
4m & \text{if $7m \le n \le 7m+2$,}\\
4m + 1  & \text{if $n = 7m+3$,}\\
4m + 2 & \text{if $n = 7m+4,7m+5$,}\\
4m + 3 & \text{if $n = 7m+6$,}\\
\end{cases}\\[1ex]
& = \floor{4n/7} - \delta(n \equiv_7 2).\qedhere
\end{align*}
\end{proof}
In~\cite{06bou}, the authors classified all optimal $[n,3]$ SO codes. It gives the minimum distances of optimal $[n,3]$ SO codes.
Based on their results, one can obtain the following theorem.
\begin{theorem}{\rm (\cite{06bou, 73BauMc})}\label{thm: optimum distance dim 3}
The following explicit formulas hold:
\begin{enumerate}[label = {\rm (\arabic*)}]
\item For $n \ge 3$,
\begin{align*}
d(n,3) &= \floor{4n/7} - \delta(n \equiv_7 2) \\[1ex]
& = \begin{cases}
\floor{4n/7} & \text{if $n \not\equiv_7 2$},\\
\floor{4n/7} - 1 & \text{if $n \equiv_7 2$}.
\end{cases}
\end{align*}
\item For $n \ge 6$,
\begin{align*}
\dso(n,3) = \begin{cases}
\floor{4n/7} & \text{if $n \equiv_7 0, 1, 5$,}\\
\floor{4n/7} - 1 & \text{if $n \equiv_7 2,3,6$,}\\
\floor{4n/7} - 2 & \text{if $n \equiv_7 4$.}
\end{cases}
\end{align*}
\end{enumerate}
\end{theorem}

\begin{proof}
For (1), by Lemma~\ref{lem: d(n,k) upper bound for dim 3}, it suffices to find an $[n,3]$ code whose minimum distance is $\floor{4n/7} - \delta(n \equiv_7 2)$.

For $3 \le t \le 9$, such codes are known in~\cite{07codetables}. For $3 \le t \le 9$, let $\C_{t}$ be an $[t,3,\floor{4t/7} - \delta(t \equiv_7 2)]$ code and $G_t$ be a generator matrix of $\C_t$. For any $n \ge 10$, we define
\[\arraycolsep=3pt
G_n := \left[ \begin{array}{c||c} s \nmH_3 & G_t \end{array} \right],
\]
where $n = 7s + t$ for $3 \le t \le 9$. Since the minimum distance of the simplex code $\calS_3$ is $4$, $G_n$ generates the $[n,3]$ code whose minimum distance is 
\[
4s + (\floor{4t/7} - \delta(n \equiv_7 2)) = \floor{4n/7} - \delta(n \equiv_7 2).
\]

For (2), reorganizing the results in~\cite[Section 3]{06bou}, we have that for $n \ge 6$,
\begin{align*}
\dso(n,3) = \begin{cases}
\floor{4n/7} & \text{if $n \equiv_7 0, 1, 5$ and $n \neq 5$,}\\
\floor{4n/7} - 1 & \text{if $n \equiv_7 2,3,6$,}\\
\floor{4n/7} - 2 & \text{if $n \equiv_7 4$ or $n = 5$.}
\end{cases}
\end{align*}
\end{proof}

From Theorem~\ref{thm: optimum distance dim 3}, we immediately obtain the following corollary.
\begin{corollary}
For $n \ge 6$, $d(n,3) = \dso(n,3)$ if and only if $n \equiv_7 0,1,2,5$.
\end{corollary}

Let us consider the four-dimensional case. The Griesmer bound says that all $[n,4,d]$ codes satisfy that 
\begin{align*}
\begin{aligned}
n &\ge d + \ceil{\frac{d}{2}} + \ceil{\frac{d}{4}} + \ceil{\frac{d}{8}} \\
&= \begin{cases}
15d/8& \text{if $d \equiv_8 0$},\\
(15d+17)/8& \text{if $d \equiv_8 1$},\\
(15d+10)/8& \text{if $d \equiv_8 2$},\\
(15d+11)/8& \text{if $d \equiv_8 3$},\\
(15d+4)/8& \text{if $d \equiv_8 4$},\\
(15d+13)/8& \text{if $d \equiv_8 5$},\\
(15d+6)/8& \text{if $d \equiv_8 6$},\\
(15d+7)/8& \text{if $d \equiv_8 7$}.
\end{cases}
\end{aligned}
\end{align*}
As the previous cases, one can obtain the following upper bound on $d(n,4)$.
\begin{lemma}{\rm (\cite{60Griesmer})}
\label{lem: d(n,k) upper bound for dim 4}
For $n\ge4$, 
\begin{align*}
d(n,4) & \le \floor{8n/15} - \delta(n \equiv_{15} 2,3,4,6,10) \\[1ex]
& = \begin{cases}
\floor{8n/15} & \text{if $n \not\equiv_{15} 2,3,4,6,10$},\\
\floor{8n/15} - 1 & \text{if $n \equiv_{15} 2,3,4,6,10$}. 
\end{cases}
\end{align*}
\end{lemma}
\begin{proof}
It can be proved in the same manner as the proof of Lemma~\ref{lem: d(n,k) upper bound for dim 3}. 
\end{proof}

\begin{theorem}{\rm (\cite{73BauMc})}\label{thm: optimum distance dim 4}
For $n\ge 4$, 
\begin{align*}
d(n,4) &= \floor{8n/15} - \delta(n \equiv_{15} 2,3,4,6,10) \\[1ex]
& = \begin{cases}
\floor{8n/15} & \text{if $n \not\equiv_{15} 2,3,4,6,10$},\\
\floor{8n/15} - 1 & \text{if $n \equiv_{15} 2,3,4,6,10$}. 
\end{cases}
\end{align*}
\end{theorem}

\begin{proof}
By Lemma~\ref{lem: d(n,k) upper bound for dim 4}, it suffices to find an $[n,4]$ code whose minimum distance is $\floor{8n/15} - \delta(n \equiv_{15} 2,3,4,6,10)$.

For $4 \le t \le 18$, such codes are known in~\cite{07codetables}. 
For $4 \le t \le 18$, let $G_t$ be a generator matrix of an optimal $[t,4]$ code.
For any $n \ge 19$, we define
\[\arraycolsep=3pt
G_n := \left[ \begin{array}{c||c} s \nmH_4 & G_t \end{array} \right],
\]
where $n = 15s + t$ for $4 \le t \le 18$. Since the minimum distance of the simplex code $\calS_4$ is $8$, $G_n$ generates the $[n,4]$ code whose minimum distance is 
\begin{align*}
&8s + (\floor{8t/15} - \delta(n \equiv_{15} 2,3,4,6,10))\\
&= \floor{8n/15} - \delta(n \equiv_{15} 2,3,4,6,10).\qedhere
\end{align*}
\end{proof}

\begin{lemma}\label{lem: no SO [n,4]}
For $n \ge 4$ such that $n\equiv_{15} 4$, there are no $\left[n,4,\floor{\frac{8n}{15}} -1 \right]$ SO codes.
\end{lemma}
\begin{proof}
Let $n = 15t + 4$ for some $t \in \Z_{\ge 0}$. Then we have
\[
\floor{\frac{8n}{15}} -1 = 8t + 1 \equiv_2 1.
\]
Since the minimum distances of SO codes should be even, our assertion follows.
\end{proof}

In~\cite{08Li}, the authors classified all optimal $[n,4]$ SO codes for $n\equiv_{15}1,2,6,7,8,9,13,14$ and proved that $\dso(n,4) < \floor{8n/15}$ for each $n \equiv_{15} 5,12$.
Combining the results in~\cite{08Li}, Lemma~\ref{lem: d(n,k) upper bound for dim 4}, and Lemma~\ref{lem: no SO [n,4]}, we derive the following theorem.

\begin{theorem}\label{thm: optimum SO distance dim 4}
For $4 \le n \le 7$, there are no $[n,4]$ SO codes and for $n \ge 8$,
{\small
\begin{align}\label{eq: optimum SO distance dim 4}
\dso(n,4)\hspace{-0.5ex} =\hspace{-0.5ex} \begin{cases} 
\floor{8n/15} & \text{\hspace{-1ex}if $n\hspace{-0.5ex}\equiv_{15}\hspace{-0.5ex} 0,1,8,9,13$ and $n\hspace{-0.5ex} \neq\hspace{-0.5ex} 13$},\\[2ex]
\floor{8n/15} - 1 & \text{\hspace{-1ex}if $n\hspace{-0.5ex}\equiv_{15}\hspace{-0.5ex} 2,3,6,7,10,11,14$}, \\[2ex]
\floor{8n/15} - 2 & \text{\hspace{-1ex}if $n\hspace{-0.5ex}\equiv_{15}\hspace{-0.5ex} 4,5,12$ or $n = 13$}.
\end{cases}
\end{align}
}
\end{theorem}
\begin{proof}
We first note that for $n\equiv_{15}1,2,6,7,8,9,13,14$, 
the equation $\dso(n,4) = \floor{\frac{8n}{15}} - \delta(n \in \calE^{(4)}_1) - 2\delta(n \in \calE^{(4)}_2)$ follows from the classification for optimal $[n,4]$ SO codes in~\cite{08Li}.

For $n \equiv_{15} 0,3,4,5,10,11,12$, combining Theorem~\ref{thm: optimum distance dim 4} with the inequality $\dso(n,4) < \floor{\frac{8n}{15}}$ from~\cite[Corollary 3.2]{08Li} gives an upper bound $\dso(n,4) \le \floor{\frac{8n}{15}} - \delta(n \in \calE^{(4)}_1) - 2\delta(n \in \calE^{(4)}_2)$.

Now, let us find an $[n,4]$ SO code whose minimum distance is equal to $\floor{\frac{8n}{15}} - \delta(n \in \calE^{(4)}_1) - 2\delta(n \in \calE^{(4)}_2)$ for $n \equiv_{15} 0,3,4,5,10,11,12$.
For $4 \le n \le 18$, optimal $[n,4]$ SO codes are given in~\cite{06bou} whose minimum distances are equal to $\floor{\frac{8n}{15}} - \delta(n \in \calE^{(4)}_1) - 2\delta(n \in \calE^{(4)}_2)$. 
For $4 \le t \le 18$, let $G^{\SO}_t$ be a generator matrix of an optimal $[t,4]$ SO code. 
For any $n = 15s + t$ when $s \ge 1$ and $t \in \{4,5,10,11,12,15,18\}$, define
\[\arraycolsep=3pt
G^{\SO}_n := \left[ \begin{array}{c||c} s \nmH_4 & G^{\SO}_t \end{array} \right].
\]
Since the minimum distance of the simplex code $\calS_4$ is $8$, $G^{\SO}_n$ generates the $[n,4]$ SO code whose minimum distance is \begin{align*}
&8s + (\floor{8t/15} - \delta(n \equiv_{15} 3,10,11) - 2\delta(n \equiv_{15} 4,5,12))\\ 
&=\floor{8n/15} - \delta(n \equiv_{15} 3,10,11)- 2\delta(n \equiv_{15} 4,5,12). \qedhere
\end{align*}
\end{proof}

From Theorem~\ref{thm: optimum distance dim 4} and~\ref{thm: optimum SO distance dim 4}, we immediately obtain the following corollary.
\begin{corollary}
For $n \ge 8$, $d(n,4) = \dso(n,4)$ if and only if $n \equiv_{15} 0,1,2,3,6,8,9,10,13$ and $n \neq 13$.
\end{corollary}

\begin{remark}\hfill
\begin{enumerate}[label = {\rm (\arabic*)}]
\item We have checked that any optimal linear code denoted by $BKLC(GF(2),n,k)$ for $k = 2,3$ and $4\le n \le 256$ in MAGMA database is embedded to an optimal self-orthogonal code using Algorithms~\ref{alg: SO embedding dim2} and~\ref{alg: SO embedding dim3}.
\item We have also checked that any optimal linear code denoted by $BKLC(GF(2),n,4)$ for $4\le n \le 256$ in MAGMA database is embedded to an optimal self-orthogonal code using the modified version of Algorithm~\ref{alg: SO embedding dim4} in the sense of Remark~\ref{rem: Alg dim 4 modify}. 
\end{enumerate}
\end{remark}

\begin{example}
Let $\C_{4,4}$ be an optimal $[4,4,1]$ code generated by
\[\arraycolsep=2.5pt
G_{4,4} = \left[
\begin{array}{ccccccccc}
0&0&0&1\\
0&0&1&0\\
0&1&0&0\\
1&0&0&0
\end{array}
\right].
\]
When we apply Algorithm~\ref{alg: SO embedding dim4} to $G_{4,4}$, we choose $s_0 = 1$ in Step ({\bf C}3) and obtain
\[\arraycolsep=2.5pt
\tG_{4,4} = \left[
\begin{array}{ccccccccc}
0&0&0&1&0&0&0&1\\
0&0&1&0&0&0&1&0\\
0&1&0&0&0&1&0&0\\
1&0&0&0&1&0&0&0
\end{array}
\right].
\]
Thus, the linear code $\tC_{4,4}$ generated by $\tG_{4,4}$ is an $[8,4,2]$ code. By Theorem~\ref{thm: optimum SO distance dim 4}, one can see that $\dso(8,4) = 4$ and thus $\tC_{4,4}$ is not an optimal SO code.

On the other hand, if we modify Algorithm~\ref{alg: SO embedding dim4} by taking $s_0 = 15$ in Step ({\bf C}3) and letting $\calI_2 = \calI_{s_0}^{(2)} \setminus J_1(\tG)$ in Step ({\bf C}7), then we obtain the new matrix
\[\arraycolsep=2.5pt
\oG_{4,4} = \left[
\begin{array}{ccccccccc}
0&0&0&1&0&1&1&1\\
0&0&1&0&1&0&1&1\\
0&1&0&0&1&1&0&1\\
1&0&0&0&1&1&1&0
\end{array}
\right].
\]
The code $\oC_{4,4}$ generated by $\oG_{4,4}$ is the well-known optimal SO code that is  the extended Hamming $[8,4,4]$ code.
\end{example}

Observing Remark V.11  and Example V.12, we conjecture the following.

\begin{conjecture}\label{conj: OSO embedding of [n,4]}
Any optimal $[n,4]$ code can be embedded to an optimal SO.
\end{conjecture}

Finally, let us consider the five-dimensional case. The Griesmer bound says that all $[n,5,d]$ codes satisfy that 
\begin{align}\label{eq: Greismer bound for dim 5}
\begin{aligned}
n &\ge d + \ceil{\frac{d}{2}} + \ceil{\frac{d}{4}} + \ceil{\frac{d}{8}}+\ceil{\frac{d}{16}} \\
&= \begin{cases}
31d/16& \text{if $d \equiv_{16} 0$},\\
(31d+49)/16& \text{if $d \equiv_{16} 1$},\\
(31d+34)/16& \text{if $d \equiv_{16} 2$},\\
(31d+35)/16& \text{if $d \equiv_{16} 3$},\\
(31d+20)/16& \text{if $d \equiv_{16} 4$},\\
(31d+37)/16& \text{if $d \equiv_{16} 5$},\\
(31d+22)/16& \text{if $d \equiv_{16} 6$},\\
(31d+ 23)/16& \text{if $d \equiv_{16} 7$},\\
(31d+8)/16& \text{if $d \equiv_{16} 8$},\\
(31d+41)/16& \text{if $d \equiv_{16} 9$},\\
(31d+26)/16& \text{if $d \equiv_{16} 10$},\\
(31d+27)/16& \text{if $d \equiv_{16} 11$},\\
(31d+12)/16& \text{if $d \equiv_{16} 12$},\\
(31d+29)/16& \text{if $d \equiv_{16} 13$},\\
(31d+14)/16& \text{if $d \equiv_{16} 14$},\\
(31d+15)/16& \text{if $d \equiv_{16} 15$}.
\end{cases}
\end{aligned}
\end{align}
As the previous cases, one can obtain the following upper bound on $d(n,5)$.

\begin{lemma}{\rm (\cite{60Griesmer})}
\label{lem: d(n,k) upper bound for dim 5}
For $n\ge5$,
\begin{align*}
d(n,5) & \le \floor{16n/31} - \delta(n \in E_1) -2 \delta(n \in E_2)\\[1ex]
& = \begin{cases}
\floor{16n/31} & \text{if $n \notin E_1 \cup E_2$},\\
\floor{16n/31} - 1 & \text{if $n \in E_1$}, \\
\floor{16n/31} - 2 & \text{if $n \in E_2 $},
\end{cases}
\end{align*}
where
\begin{align*}
E_1 &:= \left\{i  \; \middle| \; i \equiv_{31} \hspace{-1ex} \begin{array}{l} 2,3,5,6,7,8,10,11,12,\\  14,18,19,20,22,26 \end{array} \right\},\\
E_2 &:=\{i \mid i \equiv_{31} 4 \}.
\end{align*}
\end{lemma}

\begin{proof}
It can be proved in the same manner as the proof of Lemma~\ref{lem: d(n,k) upper bound for dim 3}.
\end{proof}

From the above lemma, one can derive an explicit formula for $d(n,5)$ as the following theorem.

\begin{theorem}{\rm (\cite{73BauMc})}
\label{thm: optimum distance dim 5}
For $n\ge 5$,
\begin{align*}
d(n,5) &= \floor{16n/31} - \delta(n \in \calE_1) -2 \delta(n \in \calE_2) \\[1ex]
& = \begin{cases}
\floor{16n/31} & \text{if $n \notin \calE_1 \cup \calE_2$},\\
\floor{16n/31} - 1 & \text{if $n \in \calE_1$}, \\
\floor{16n/31} - 2 & \text{if $n \in \calE_2 $}, 
\end{cases}
\end{align*}
where
$\calE_1 := E_1 \cup \{9,13\} \setminus \{8,12\}$ and 
$\calE_2 := E_2 \cup \{8,12\}$.
\end{theorem}

\begin{proof}
It can be proved in the same manner as the proof of Theorem \ref{thm: optimum distance dim 4}.
\end{proof}

To obtain an explicit formula for $\dso(n,5)$, we prove the following lemma.

\begin{lemma}\label{lem: No SO dim 5}
For $n \ge 5$, if
\begin{enumerate}[label = {\rm (\arabic*)}]
\item $n = 13$,
\item $n \equiv_{31} 12$ and $n \neq 12$, or
\item $n \equiv_{31} 5,8,15,20,23,27,30$,
\end{enumerate}
then there are no $[n,5,d(n,5)]$ SO codes.
\end{lemma}
\begin{proof}
By Theorem~\ref{thm: optimum distance dim 5}, $d(n,5)$ is odd if $n$ satisfies either (1), (2), or (3) for $n \ge 5$.
Since the minimum distances of SO codes should be even, our assertion follows.
\end{proof}

Now, we obtain the following theorem.

\begin{theorem}\label{thm: optimum SO distance dim 5}
For $n \ge 11$, if $n = 13$ or $n \not\equiv_{31} 6,13,14,21,22,28,29$, then
\begin{align}\label{eq: optimum SO distance dim 5}
\begin{aligned}
\dso(n,5) &= \floor{16n/31} - \delta(n \in \calE^{\SO}_1) -2 \delta(n \in \calE^{\SO}_2) \\[1ex]
& = \begin{cases}
\floor{16n/31} & \text{if $n \notin \calE^{\SO}_1 \cup \calE^{\SO}_2$},\\
\floor{16n/31} - 1 & \text{if $n \in \calE^{\SO}_1$}, \\
\floor{16n/31} - 2 & \text{if $n \in \calE^{\SO}_2 $}, 
\end{cases}
\end{aligned}
\end{align}
where
\begin{align*}
\calE^{\SO}_1 &:= \{i \mid i \equiv_{31} 2,3,7,10,11,15,18,19,23,26,27,30 \},\\
\calE^{\SO}_2 &:= \{i \mid i \equiv_{31} 4,5,8,12,20\} \cup \{13\}.
\end{align*}
\end{theorem}

\begin{proof}
Combining Theorem~\ref{thm: optimum distance dim 5} with Lemma~\ref{lem: No SO dim 5} gives an upper bound $\dso(n,5) \le \floor{\frac{16n}{31}} - \delta(n \in \calE^{(5)}_1) -2 \delta(n \in \calE^{(5)}_2)$.
Therefore, it suffices to find an $[n,5]$ SO code whose minimum distance is equal to the right-hand side of Equation~\eqref{eq: optimum SO distance dim 5}.

For $11 \le t \le 40$, optimal SO codes are given in~\cite{06bou}. Their minimum distances are equal to the right-hand side of Equation~\eqref{eq: optimum SO distance dim 5}. 
For $10 \le t \le 40$, let $G^{\SO}_t$ be a generator matrix of an optimal $[t,5]$ SO code.

For any $n = 31s + t$ when $s \ge 1$ and $t \in \{i \in \{11,12,\ldots, 41\} \mid   i \neq 13,14,21,22,28,29,37\}$, we define
\[\arraycolsep=3pt
G^{\SO}_n := \left[ \begin{array}{c||c} s \nmH_5 & G^{\SO}_t \end{array} \right].
\]
Since the minimum distance of the simplex code $\calS_5$ is $15$, $G^{\SO}_n$ generates the $[n,5]$ SO code whose minimum distance is 
\begin{align*}
&15s + (\floor{16t/31} - \delta(n \in \calE^{\SO}_1) -2 \delta(n \in \calE^{\SO}_2))\\ 
&=\floor{16n/31} - \delta(n \in \calE^{\SO}_1) -2 \delta(n \in \calE^{\SO}_2). \qedhere
\end{align*}

\end{proof}

\begin{remark}\label{rem: dso(n,5)}\hfill
\begin{enumerate}[label = {\rm (\arabic*)}]
\item Combining Theorems~\ref{thm: optimum distance dim 5} and~\ref{thm: optimum SO distance dim 5}, one can see that  there are $[n-1,5,d(n,5)-2]$ SO codes if $n\equiv_{31} 6,13,14,21,22,28,29$ for $n\ge 10$. Therefore, one can obtain an $[n,5,d(n,5)-2]$ SO code by attaching a zero coordinate to an $[n-1,5,d(n,5)-2]$  SO code.

\item  For $10 \le n \le 100$, $n \equiv_{31} 6,13,14,21,22,28,29$, and $n \neq 13$,  we choose $2^{16}$ random $[n,5]$ SO codes and calculate their minimum distances using MAGMA~\cite{97MAGMA}. In our calculation, there was no SO code whose minimum distance is equal to $d(n,5)$.
\end{enumerate}
\end{remark}

Based on Table 1 in~\cite{06bou} and Remark~\ref{rem: dso(n,5)}, we leave the following conjecture. 

\begin{conjecture}\label{conj: dso(n,5)}
For $n\ge 10$, if $n\neq13$ and $n\equiv_{31} 6,13,14,21,22,28,29$, then
\[
\dso(n,5) = d(n,5) - 2,
\]
that is, there are no $[n,5,d(n,5)]$ SO codes.
\end{conjecture}

\section{Conclusion}

We have introduced a characterization on self-orthogonality for given binary linear codes in terms of the number of column vectors in its generator matrix. In particular, we have described the characterization explicitly  for each $k = 2,3,4$.

As a consequence of our characterizations, for each $k=2,3,4$, we have proposed  algorithms that embed an $[n,k]$ code to a self-orthogonal code by minimum lengthening.  We have also suggested an algorithm that embeds an $[n,k]$ code to a self-orthogonal code for $k >4$. 

We have also given new explicit formulas for the minimum distances of optimal linear codes for dimensions 4 and 5 and those of optimal self-orthogonal codes for any length $n$ with dimension 4 and any length $n \not\equiv 6,13,14,21,22,28,29 \pmod{31}$ with dimension 5.

Using our explicit formulas and MAGMA, we have obtained that the above algorithms embed optimal linear codes into optimal self-orthogonal codes for $n \le 256$ and $k=2,3,4$.

Finally, we have suggested two conjectures in Conjecture~\ref{conj: OSO embedding of [n,4]} and Conjecture~\ref{conj: dso(n,5)}.

\appendix
\section{Appendix}
\subsection{Proof of Theorem~\ref{thm: equiv cond for SO dim 4}}\label{apx: Proof of Thm 3.8}
\begin{proof}
$(\Leftarrow)$ We prove the case where $s = 1$ since the other cases can be easily shown in the same manner.
By the hypothesis, for each $t = 1,2$, $\ell_i \equiv_2 \ell_j$ for $i,j \in \calI_1^{(t)}$, that is,
{\small
\begin{align}
\begin{aligned}\label{eq: I_1 condition}
&\ell_1 \equiv_2 \ell_2  \equiv_2 \ell_3  \equiv_2 \ell_4  \equiv_2 \ell_5  \equiv_2 \ell_6  \equiv_2 \ell_7 \quad \text{and}\\
&\ell_8  \equiv_2 \ell_9  \equiv_2 \ell_{10}  \equiv_2 \ell_{11}  \equiv_2 \ell_{12}  \equiv_2 \ell_{13}  \equiv_2 \ell_{14}  \equiv_2 \ell_{15}.
\end{aligned}
\end{align}
}
Combining Equations~\eqref{eq: inner product of rows dim 4} and~\eqref{eq: I_1 condition}, we have
\begin{align*}
&|\bfI(1)| = \ell_{1} + \ell_{3} + \ell_{5} + \ell_{7} + \ell_{9} + \ell_{11} + \ell_{13} + \ell_{15} \equiv_2 0, \\
&|\bfI(2)| = \ell_{2} + \ell_{3} + \ell_{6} + \ell_{7} + \ell_{10} + \ell_{11} + \ell_{14} + \ell_{15} \equiv_2 0, \\
&|\bfI(3)| = \ell_{4} + \ell_{5} + \ell_{6} + \ell_{7} + \ell_{12} + \ell_{13} + \ell_{14} + \ell_{15} \equiv_2 0, \\
&|\bfI(4)| = \ell_{8} + \ell_{9} + \ell_{10} + \ell_{11} + \ell_{12} + \ell_{13} + \ell_{14} + \ell_{15} \equiv_2 0, \\
&|\bfI(1) \cap \bfI(2)| = \ell_{3} + \ell_{7} + \ell_{11} + \ell_{15} \equiv_2 0,\\
&|\bfI(1) \cap \bfI(3)| = \ell_{5} + \ell_{7} + \ell_{13} + \ell_{15} \equiv_2 0, \\
&|\bfI(1) \cap \bfI(4)| = \ell_{9} + \ell_{11} + \ell_{13} + \ell_{15} \equiv_2 0, \\
&|\bfI(2) \cap \bfI(3)| = \ell_{6} + \ell_{7} + \ell_{14} + \ell_{15} \equiv_2 0, \\
&|\bfI(2) \cap \bfI(4)| = \ell_{10} + \ell_{11} + \ell_{14} + \ell_{15} \equiv_2 0, \\
&|\bfI(3) \cap \bfI(4)| = \ell_{12} + \ell_{13} + \ell_{14} + \ell_{15} \equiv_2 0.
\end{align*}
Therefore,  $|\bfI(j) \cap \bfI(j')|$ is even for all $0 < j \le j' \le 4$ and hence $\C$ is SO by Theorem~\ref{thm: eq cond for SO}.

\noindent $(\Rightarrow)$ 
For $j = 0,1$, let
\[
J_j(G) := \left\{i \in \{1,2,\ldots,15\} \; \middle| \; \ell_i \equiv_2 j  \right\}.
\]
Note that one of $|J_0(G)|$ and $|J_1(G)|$ is larger than or equal to $8$.
Thus, we can choose $I := \{i_1 < i_2 < \cdots < i_8\} \subset \{1,2,\ldots, 15\}$ so that 
\begin{align}\label{eq: i_k condition}
\ell_{i_k} \equiv_2 \ell_{i_{k'}} \quad \text{for $k,k' =1,2,\ldots,8$.}
\end{align}
Note that there exists a subset $S := \{j_1,j_2,j_3,j_4\}$ of $\{i_1,i_2,\ldots,i_8\}$ such that
\[
S \subset \{1,2,\ldots, 7\} \quad \text{or} \quad  S \subset \{8,9,\ldots,15\}.
\]

Suppose that $S \subset \{1,2,\ldots, 7\}$.

\begin{claim}\label{claim 1}
If there is $s \in \{1,2,\ldots, 7\}$ such that 
\[
\calP_s^{(1)} = \{(j_{\sigma(1)},j_{\sigma(2)}), (j_{\sigma(3)}, i), (j_{\sigma(4)},i')\}
\]
for a permutation $\sigma$ of $\{1,2,3,4\}$ and $i,i' \in \{1,2,\ldots,7\}$, then
\[
\ell_1 \equiv_2 \ell_2 \equiv_2 \ell_3 \equiv_2 \ell_4 \equiv_2 \ell_5 \equiv_2 \ell_6 \equiv_2 \ell_7.
\]
Note that we consider $(i,j)$ and $(j,i)$ the same.
\end{claim}

\begin{claim}\label{claim 2}
If there are no $s \in \{1,2,\ldots, 7\}$ such that 
\[
\calP_s^{(1)} = \{(j_{\sigma(1)},j_{\sigma(2)}), (j_{\sigma(3)}, i), (j_{\sigma(4)},i')\}
\]
for any permutation $\sigma$ of $\{1,2,3,4\}$ and $i,i' \in \{1,2,\ldots,7\}$, then $S$ should be one of the following sets:
\[
\begin{array}{cccc}
\{1,2,4,7\}, & \{1,2,5,6\}, & \{1,3,4,6\}, & \{1,3,5,7\},\\
\{2,3,4,5\}, & \{2,3,6,7\}, & \{4,5,6,7\}.
\end{array}
\]
\end{claim}

For Claim~\ref{claim 1}, suppose that there is $s \in \{1,2,\ldots, 7\}$ such that
\[
\calP_s^{(1)} = \{(j_{\sigma(1)},j_{\sigma(2)}), (j_{\sigma(3)}, i), (j_{\sigma(4)},i')\}
\]
for a permutation $\sigma$ of $\{1,2,3,4\}$ and $i,i' \in \{1,2,\ldots,7\}$. Since $\C$ is SO, we have the following by Lemma~\ref{lem: equiv cond for SO dim 4}. 
\[
\ell_{j_{\sigma(1)}} + \ell_{j_{\sigma(2)}} \equiv_2 \ell_{j_{\sigma(3)}} + \ell_{i} \equiv_2 \ell_{j_{\sigma(4)}} + \ell_{i'}.
\]
The following equivalence is obtained by the condition~\eqref{eq: i_k condition} since $S \subset \{i_1,i_2,\ldots,i_8\}$.
\begin{align}\label{eq: jks i i' cong}
\ell_{j_1} \equiv_2 \ell_{j_2} \equiv_2 \ell_{j_3} \equiv_2 \ell_{j_4} \equiv_2 \ell_{i} \equiv_2 \ell_{i'}.
\end{align}
Since $j_1 \in \{1,2,\ldots,7\}$, we attain the following by Remark~\ref{rem: prop of calP} (1).
\begin{align*}
\{i \in \Z \mid \text{$i$ appears in $\calP_{j_1}^{(1)}$}\} &= \{1,2,\ldots,7\} \setminus \{j_1\}\\
& = \{s,j_2,j_3,j_4,i,i'\}.
\end{align*}
Thus, by Lemma~\ref{lem: equiv cond for SO dim 4} and Equation~\eqref{eq: jks i i' cong}, we have
\[
\ell_1 \equiv_2 \ell_2 \equiv_2 \ell_3 \equiv_2 \ell_4 \equiv_2 \ell_5 \equiv_2 \ell_6 \equiv_2 \ell_7.
\]

For Claim~\ref{claim 2}, suppose that there are no $s \in \{1,2,\ldots, 7\}$ such that  
\[
\calP_s^{(1)} = \{(j_{\sigma(1)},j_{\sigma(2)}), (j_{\sigma(3)}, i), (j_{\sigma(4)},i')\}
\]
for any permutation $\sigma$ of $\{1,2,3,4\}$ and $i,i' \in \{1,2,\ldots,7\}$. Then for all $s \in \{1,2,\ldots,7\}\setminus S$,
\[
\calP_s^{(1)} = \{(j_{\sigma(1)},j_{\sigma(2)}), (j_{\sigma(3)}, j_{\sigma(4)}), (i,i')\},
\]
for a permutation $\sigma$ of $\{1,2,3,4\}$ and $i,i'\in \{1,2,\ldots,7\}\setminus S$. Considering Table~\ref{table: calP}, one can see that $S$ should be one of the following sets:
\begin{align}\label{eq: sets for claim 2}\arraycolsep=3pt
\begin{array}{llll}
\{1,2,4,7\}, & \{1,2,5,6\}, & \{1,3,4,6\}, & \{1,3,5,7\},\\
\{2,3,4,5\}, & \{2,3,6,7\}, & \{4,5,6,7\}.
\end{array}
\end{align}

For the case where $S \subset \{8,9,\ldots, 15\}$, we have similar claims as follows.

\begin{claim}\label{claim 3}
If there is $s \in \{1,2,\ldots, 7\}$ such that 
\[
\calP_s^{(2)} = \{(j_{\sigma(1)},j_{\sigma(2)}), (j_{\sigma(3)}, i), (j_{\sigma(4)},i'),(\alpha,\beta)\}
\]
for a permutation $\sigma$ of $\{1,2,3,4\}$ and $i,i',\alpha,\beta \in \{8,9,\ldots,15\}$, then
\[
\ell_8 \equiv_2 \ell_9 \equiv_2 \ell_{10} \equiv_2 \ell_{11} \equiv_2 \ell_{12} \equiv_2 \ell_{13} \equiv_2 \ell_{14} \equiv_2 \ell_{15}.
\]
\end{claim}

\begin{claim}\label{claim 4}
If there are no $s \in \{1,2,\ldots, 7\}$ such that 
\[
\calP_s^{(2)} = \{(j_{\sigma(1)},j_{\sigma(2)}), (j_{\sigma(3)}, i), (j_{\sigma(4)},i'),(\alpha,\beta)\}
\]
for any permutation $\sigma$ of $\{1,2,3,4\}$ and $i,i',\alpha,\beta \in \{8,9,\ldots,15\}$, then $S$ should be one of the following sets:
\begin{align}\label{eq: sets for claim 4}
\arraycolsep=2.5pt
{\small\begin{array}{lll}
\{8,9,10,11\}, & \{8,9,12,13\}, & \{8,9,14,15\},\\
\{8,10,12,14\},& \{8,10,13,15\}, & \{8,11,12,15\}, \\
\{8,11,13,14\}, & \{9,10,12,15\}, & \{9,10,13,14\}, \\
\{9,11,12,14\}, & \{9,11,13,15\}, & \{9,12,13,14\}, \\
\{10,11,12,13\}, & \{10,11,14,15\}, & \{12,13,14,15\}.
\end{array}}
\end{align}
\end{claim}

Claims~\ref{claim 3} and~\ref{claim 4} can be proved in the same manner as the proof of Claims~\ref{claim 1} and~\ref{claim 2}. Therefore, we omit their proofs.	

Now we consider the following four cases.\smallskip

\noindent {\it Case 1: $\left|I \cap \{1,2,\ldots,7\}\right| > 4$.}

In this case, there should be $\{j_1,j_2,j_3,j_4\} \subset I$ and $s \in \{1,2,\ldots,7\}$ such that 
\[
\calP_s^{(1)} = \{(j_1,j_2), (j_3, i), (j_4,i')\}
\]
for some $i,i' \in \{1,2,\ldots,7\}$. 
Thus, we have the following by Claim~\ref{claim 1}. 
\begin{align}\label{eq: all ell equal 1}
\ell_1 \equiv_2 \ell_2 \equiv_2 \cdots \equiv_2 \ell_7.
\end{align}
For any $j \in \{8,9,\ldots, 15\}$, there exists $s \in \{1,2,\ldots,7\}$ such that $(i_8,j) \in \calP_s^{(2)}$ by Remark~\ref{rem: prop of calP} (2) since $i_8$ should be in $\{8,9,\ldots, 15\}$. Therefore, by Lemma~\ref{lem: equiv cond for SO dim 4} and Equation~\eqref{eq: all ell equal 1}, we have
\[
\ell_{i_8} + \ell_j \equiv_2 \ell_{h_1} + \ell_{h_2} \equiv_2 0
\] 
for any $(h_1,h_2) \in \calP_s^{(1)}$.
Thus, we have
\[
\ell_1 \equiv_2 \ell_2 \equiv_2 \cdots \equiv_2 \ell_{15}.
\]

\noindent {\it Case 2: $0 < \left|I \cap \{1,2,\ldots,7\}\right| < 4$.}

In this case, one can show 
\[
\ell_1 \equiv_2 \ell_2 \equiv_2 \cdots \equiv_2 \ell_{15}
\] 
in the same manner as Case 1. \smallskip

\noindent {\it Case 3: $\left|I \cap \{1,2,\ldots,7\}\right| = 0$, that is, $I = \{8,9,\ldots,15\}$.}

Suppose that there is $i \in \{1,2,\ldots, 7\}$ such that $\ell_i \equiv_2 \ell_{i_1}$. Let
\[
I' := \{i, i_2, i_3,\ldots, i_8\}.
\]
We obtain the following by Case 2 since $0< \left|I' \cap \{1,2,\ldots, 7\} \right| = 1 <4$. 
\[
\ell_1 \equiv_2 \ell_2 \equiv_2 \cdots \equiv_2 \ell_{15}.
\]
If $\ell_i \not \equiv_2 \ell_{i_1}$ for all $i \in \{1,2,\ldots, 7\}$, then for $t = 1,2$, we have
\begin{align*}
&\ell_{1} \equiv_2 \ell_{2} \equiv_2 \ell_{3} \equiv_2 \ell_{4} \equiv_2 \ell_{5} \equiv_2 \ell_{6} \equiv_2 \ell_{7} \quad \text{and}\\
&\ell_{8} \equiv_2 \ell_{9} \equiv_2 \ell_{10} \equiv_2 \ell_{11} \equiv_2 \ell_{12} \equiv_2 \ell_{13} \equiv_2 \ell_{14} \equiv_2 \ell_{15},
\end{align*}
that is,
\[
\ell_i \equiv_2 \ell_j \quad \text{for $i,j \in \calI_1^{(t)}$}.
\]
\noindent {\it Case 4: $\left| I \cap \{1,2,\ldots,7 \} \right| = 4$, that is
\begin{align*}
&\{i_1,i_2,i_3,i_4\} \subset \{1,2,\ldots, 7\}\quad \text{and}\\
&\{i_5,i_6,i_7,i_8\} \subset \{8,9,\ldots, 15\}.
\end{align*}}
If there is $i \in \{1,2,\ldots, 15\} \setminus I$ such that $\ell_i \equiv_2 \ell_{i_1}$, then by letting $I' = \{i, i_2, i_3,\ldots, i_8\}$, we obtain
\[
\ell_1 \equiv_2 \ell_2 \equiv_2 \cdots \equiv_2 \ell_{15}.
\]
Assume that 
\begin{align}\label{eq: assumption in thm}
\ell_i \not \equiv_2 \ell_{i_1} \quad \text{for all $i \in \{1,2,\ldots, 15\} \setminus I$}.
\end{align}
Suppose that there is $s\in \{1,2,\ldots, 7\}$ such that 
\begin{align*}
&\calP_s^{(1)} = \{(i_{\sigma(1)},i_{\sigma(2)}), (i_{\sigma(3)}, j), (i_{\sigma(4)},j')\} \quad \text{or} \\
&\calP_s^{(2)} = \{(i_{\sigma(1)},i_{\sigma(2)}), (i_{\sigma(3)}, h), (i_{\sigma(4)},h'),(\alpha,\beta)\}
\end{align*}
for a permutation $\sigma$ of $\{1,2,3,4\}$, $j,j' \in \{1,2,\ldots, 7\}$, and $h,h',\alpha,\beta \in \{8,9,\ldots,15\}$.
Then, by Claims~\ref{claim 1} and~\ref{claim 3}, we have
\[
\ell_1 \equiv_2 \ell_2 \equiv_2 \cdots \equiv_2 \ell_7 \quad \text{or} \quad \ell_8 \equiv_2 \ell_9 \equiv_2 \cdots \equiv_2 \ell_{15},
\]
which contradicts to the assumption~\eqref{eq: assumption in thm}.
Thus, by Claims~\ref{claim 2} and~\ref{claim 4}, $\{i_1,i_2,i_3,i_4\}$ is one of the sets in~\eqref{eq: sets for claim 2} and $\{i_5,i_6,i_7,i_8\}$ is one of the sets in~\eqref{eq: sets for claim 4}.

Now, we need to calculate all these cases. For instance, suppose that
\[
\{i_1,i_2,i_3,i_4\} = \{1,2,4,7\} ~\text{and}~ \{i_5,i_6,i_7,i_8\} = \{8,9,10,11\}.
\]
By Lemma~\ref{lem: equiv cond for SO dim 4}, we have $\ell_{3} \equiv_2 \ell_{i}$ for all $i \in I$ since $\calP_8^{(1)} = \{(1,9),(2,10),(3,11)\}$.
Moreover, since 
\begin{align*}
\calP_{12}^{(1)} &= \{(1,13),(2,14),(3,15)\}\quad \text{and}\\ 
\calP_{12}^{(2)} &= \{(4,8),(5,9),(6,10),(7,11)\}
\end{align*}
for any $i \in I \cup \{3\}$, we have
\begin{align}\label{eq: ell equiv ex}
\ell_{5} \equiv_2 \ell_{6} \equiv_2 \ell_{13} \equiv_2 \ell_{14} \equiv_2 \ell_{15} \equiv_2 \ell_i.
\end{align}
Since $\calP_{8}^{(2)} = \{(4,12),(5,13),(6,14),(7,15) \}$, we have 
\[
\ell_1 \equiv_2 \ell_2 \equiv_2 \cdots \equiv_2 \ell_{15}.
\]
by Lemma~\ref{lem: equiv cond for SO dim 4} and Equation~\eqref{eq: ell equiv ex}.

In the same manner, one can show that the other cases induce
$\ell_1 \equiv_2 \ell_2 \equiv_2 \cdots \equiv_2 \ell_{15}$
except for the following cases:
\[
\begin{array}{c|c}
\{i_1,i_2,i_3,i_4\} 	& \{i_5,i_6,i_7,i_8\}  \\ \hline\hline
\{4,5,6,7\}	& \{12,13,14,15\}	\\ \hline
\{4,5,6,7\}	& \{8,9,10,11\} 	\\ \hline
\{2,3,6,7\}	& \{10,11,14,15\} \\ \hline
\{2,3,6,7\}	& \{8,9,12,13\} 	\\ \hline
\{2,3,4,5\}	& \{10,11,12,13\}	\\ \hline
\{2,3,4,5\}	& \{8,9,14,15\} 	\\ \hline
\{1,3,5,7\}	& \{9,11,13,15\}	\\ \hline
\{1,3,5,7\}	& \{8,10,12,14\} 	\\ \hline
\{1,3,4,6\}	& \{9,11,12,14\} 	\\ \hline
\{1,3,4,6\}	& \{8,10,13,15\} 	\\ \hline
\{1,2,5,6\}	& \{9,10,13,14\} 	\\ \hline
\{1,2,5,6\}	& \{8,11,12,15\} 	\\ \hline
\{1,2,4,7\}	& \{9,10,12,15\} 	\\ \hline
\{1,2,4,7\}	& \{8,11,13,14\}
\end{array}
\]
Thus, we have
\[
\ell_i \equiv_2 \ell_j \quad \text{for $(i,j \in I)$ or $(i,j \in \{1,2,\ldots,15\}\setminus I)$}.
\]
Note that for each exceptional case, there is $s \in \{2,3,\ldots,15\}$ such that
\[
\calI_s^{(1)} = \{1,2,\ldots,15\}\setminus I \quad \text{and} \quad \calI_s^{(2)} = I.
\] 
Hence, our assertion holds.
\end{proof}

\subsection{Proof of Theorem~\ref{thm: shortest alg dim 4}}\label{apx: Proof of 4.12}

\begin{proof}
One can easily see that if $\C$ is SO, we have $\fn^{(1)}_{s_0} + \fn^{(2)}_{s_0} = 0$ by Theorem~\ref{thm: equiv cond for SO dim 4}. This implies that $\tG = G$ and thus $\tC = \C$. 
Therefore, we may assume that $\C$ is not SO.

Suppose that $s_0$ is chosen in Step ({\bf C}3). One can see that Steps ({\bf C}5) and ({\bf C}6) (resp. Steps ({\bf C}7) and ({\bf C}8)) will  run repeatedly until
$\calI_{s_0}^{(t)} \cap J_1(\tG) = \emptyset$ or
$\calI_{s_0}^{(t)} \subset J_1(\tG)$ for $t=1$ (resp. $t=2$). Therefore, $s_0$ and $\tG$ satisfy Equation~\eqref{eq: cong eqn for SO dim 4} and thus $\tC$ is SO by Theorem~\ref{thm: equiv cond for SO dim 4}.

Let $\fn$ be the difference between the length of $\tC$ and $\C$,
\[\arraycolsep=2.5pt
\hG : = \left[\begin{array}{c||c} G & M \end{array}\right],
\]
where $M$ is a $4 \times l$ matrix for some $0< l <\fn$ and $\hC$ is a linear code generated by $\hG$.

Note that, by Theorem~\ref{thm: equiv cond for SO dim 4}, $\hC$ is SO if and only if there is $s \in \{1,2,\ldots, 15\}$ such that
\begin{align}\label{eq: SO condtion for embedding}
\calI_s^{(t)} \cap J_1(\hG) = \emptyset \quad \text{or} \quad \calI_s^{(t)} \cap J_1(\hG) = \calI_s^{(t)}
\end{align}
for all $t = 1,2$. 

For $s \in \{1,2,\ldots, 15\}$, let $M^{(1)}_s$ and $M^{(2)}_s$ be submatrices of $M$ such that 
$M^{(1)}_s$ (resp. $M^{(2)}_s$) consists of column vectors $\sfh_{i}$'s where $i \in \calI_s^{(1)}$ (resp. $i \in \calI_s^{(2)}$). There is an $n \times n$ permutation matrix $P$ such that
\[\arraycolsep=2.5pt
MP  = \left[ \begin{array}{c|c} M^{(1)}_s &  M^{(2)}_s \end{array} \right],
\]
that is, $\arraycolsep=2.5pt \left[ \begin{array}{c|c} M^{(1)}_s &  M^{(2)}_s \end{array} \right]$ is $M$ with the columns interchanged.

For $t = 1,2$, let 
\[
l_s^{(t)} := \left(\text{the number of columns of $M^{(t)}_s$}\right).
\]
We also let
\begin{align}\label{eq: fn(1) def}
\fn^{(1)}_{s,0} := |\calI_s^{(1)} \setminus J_1(G)|\;~ \text{and} \;~ \fn^{(1)}_{s,1} := |\calI_s^{(1)} \cap J_1(G)|
\end{align}
and
\begin{align}
\label{eq: fn(2) def}
\fn^{(2)}_{s,0} := |\calI_s^{(2)} \setminus J_1(G)| \;~ \text{and} \;~ \fn^{(2)}_{s,1} := |\calI_s^{(2)} \cap J_1(G)|.
\end{align}
For $t = 1,2$, take $j^{(t)} \neq j'^{(t)} \in \{0,1\}$ so that $\fn^{(t)}_{s} = \fn^{(t)}_{s,j^{(t)}}$, where $\fn^{(1)}_{s}$ and $\fn^{(2)}_{s}$ are the integers defined in Steps ({\bf C}1) and ({\bf C}2), respectively. 

Since $l_s^{(1)} + l_s^{(2)} = l < \fn = \mathrm{min}\{\fn^{(1)}_s + \fn^{(2)}_s \mid 1 \le s \le 15\}$, there is $t_0 \in\{1,2\}$ such that 
\[
l_s^{(t_0)} < \fn^{(t_0)}_s \left( = \fn^{(t_0)}_{s,j^{(t_0)}} \right).
\]
In case where $t_0 = 1$, since $l_s^{(1)} < \fn_{s,j^{(1)}}^{(1)}$ and $\fn_{s,j^{(1)}}^{(1)} + \fn_{s,j'^{(1)}}^{(1)} = 7$, we have
\[
0 < \fn_{s,j^{(1)}}^{(1)} - l_s^{(1)} \le |\calI_s^{(1)} \cap J_1(\hG)| \le \fn_{s,j'^{(1)}}^{(1)} + l_s^{(1)} < 7,
\]
by Equation~\eqref{eq: fn(1) def}. Since $|\calI_s^{(1)}| = 7$,  $\hC$ is not SO by Equation~\eqref{eq: SO condtion for embedding}.

Similarly, in case where $t_0 = 2$, since $l_s^{(2)} < \fn_{s,j^{(2)}}^{(2)}$ and $\fn_{s,j^{(2)}}^{(2)} + \fn_{s,j'^{(2)}}^{(2)} = 8$, we have
\[
0 < \fn_{s,j^{(2)}}^{(2)} - l_s^{(2)} \le |\calI_s^{(2)} \cap J_1(\hG)| \le \fn_{s,j'^{(2)}}^{(2)} + l_s^{(2)} < 8,
\]
by Equation~\eqref{eq: fn(2) def}. Since $|\calI_s^{(2)}| = 8$, $\hC$ is not SO by Equation~\eqref{eq: SO condtion for embedding}.
\end{proof}

\subsection{ Algorithm~\ref{alg: SO embedding dim4} in MAGMA: Construction of a shortest SO embedding  for dimension four }\label{apx: alg dim 4}
{\tt \scriptsize 
\noindent /*\\
h{\char95}vector(i) gives the $i$th column of the generator matrix $H_k$ of the $[2^k-1,k]$ simplex code.\\
Input: The dimension $k$ and a column index $1 \le i \le 2^k-1$\\
Output: the $i$th column vector of $H_k$\\
*/\\
\noindent function h{\char95}vector(k,i)\\
\indent H{\char95}k := ZeroMatrix(IntegerRing(),k,2{\char94} k-1);\\
\indent for i in [1..2{\char94}k-1] do \\
\indent \indent for j in [0..k-1] do \\
\indent \indent \indent if Floor(i/2{\char94}j) mod 2 eq 1 then\\
\indent \indent \indent \indent H{\char95}k[k-j,i] := 1;\\
\indent \indent \indent end if;\\
\indent \indent end for;\\
\indent end for;\\	
\indent return ColumnSubmatrix(H{\char95}k,i,1);\\
end function;\\
/*\\
Num{\char95}cols(G) gives the list of $\ell_i(G)$.\\
Input: a generator matrix $G$ of $[n,k]$ code\\
Output: the list of $\ell_i(G)$'s\\
*/\\
\noindent function Num{\char95}cols(G)\\
\indent k := Nrows(G); col{\char95}mult{\char95}set := \{**\}; ell{\char95}s := []; \\
\indent for j in [1..Ncols(G)] do\\
\indent \indent ind := 0;\\
\indent \indent for i in [1..k] do\\
\indent \indent \indent ind := ind + G[i][j]*2{\char94}(k-i);\\
\indent \indent end for;\\
\indent Include({\char126}col{\char95}mult{\char95}set,ind);\\
\indent end for;\\
\indent for s in [1..2{\char94}k - 1] do\\
\indent \indent Append({\char126}ell{\char95}s,Multiplicity(col{\char95}mult{\char95}set,s));\\
\indent end for;\\
\indent return ell{\char95}s;\\
end function;\\
/*\\
Sets in Table~\ref{table: calI}.\\
*/\\
I1 := \{@\\
\{1,2,3,4,5,6,7\}, 
\{1,2,3,8,9,10,11\}, \\
\{1,2,3,12,13,14,15\}, 
\{1,4,5,8,9,12,13\}, \\
\{1,4,5,10,11,14,15\}, 
\{1,6,7,8,9,14,15\}, \\
\{1,6,7,10,11,12,13\}, 
\{2,4,6,8,10,12,14\}, \\
\{2,4,6,9,11,13,15\}, 
\{2,5,7,8,10,13,15\}, \\
\{2,5,7,9,11,12,14\}, 
\{3,4,7,8,11,12,15\}, \\
\{3,4,7,9,10,13,14\}, 
\{3,5,6,8,11,13,14\}, \\
\{3,5,6,9,10,12,15\} @\}; \\
I2 := \{@ \\
\{8,9,10,11,12,13,14,15\}, 
\{4,5,6,7,12,13,14,15\}, \\
\{4,5,6,7,8,9,10,11\}, 
\{2,3,6,7,10,11,14,15\}, \\
\{2,3,6,7,8,9,12,13\}, 
\{2,3,4,5,10,11,12,13\}, \\
\{2,3,4,5,8,9,14,15\}, 
\{1,3,5,7,9,11,13,15\}, \\ 
\{1,3,5,7,8,10,12,14\}, 
\{1,3,4,6,9,11,12,14\}, \\ 
\{1,3,4,6,8,10,13,15\}, 
\{1,2,5,6,9,10,13,14\}, \\
\{1,2,5,6,8,11,12,15\}, 
\{1,2,4,7,9,10,12,15\}, \\
\{1,2,4,7,8,11,13,14\} @\}; \\
/*\\
J1(G) gives the set $J_1(G)$ defined in Subsection 4.B.\\
Input: a generator matrix $G$ of an $[n,4]$ code.\\
Output: the set $J_1(G)$\\
*/\\
function J1(G)\\
\indent ell{\char95}i{\char95}set := Num{\char95}cols(G); J{\char95}1 := \{@@\};\\
\indent for k in [1..15] do\\
\indent \indent if IsEven(ell{\char95}i{\char95}set[k]) eq false then\\
\indent \indent \indent Include({\char126}J{\char95}1,k); \\
\indent \indent end if;\\
\indent end for;\\
\indent return J{\char95}1;\\
end function;\\
/* \\
SOconst{\char95}matrix{\char95}dim4(G) gives a generator matrix for a shortest SO embedding of an $[n,4]$ linear code (Algorithm~\ref{alg: SO embedding dim4}).\\
Input: A generator matrix $G$ of an $[n,4]$ code\\
Output: A generator matrix $\widetilde{G}$ for a shortest SO embedding\\
*/\\
function SOconst{\char95}matrix{\char95}dim4(G)\\
\indent BR := BaseRing(G); G := Matrix(IntegerRing(),G);\\
//(C1) and (C2)\\
\indent n1 := []; n2 := []; J{\char95}1 := J1(G);\\
\indent for s in [1..15] do \\
\indent \indent if \#(I1[s] meet J{\char95}1) lt 4 then\\
\indent \indent \indent Append({\char126}n1,\#(I1[s] meet J{\char95}1));\\
\indent \indent else\\
\indent \indent \indent Append({\char126}n1,\#(I1[s] diff J{\char95}1));\\
\indent \indent end if;\\
\indent \indent if \#(I2[s] meet J{\char95}1) le 4 then\\
\indent \indent \indent Append({\char126}n2,\#(I2[s] meet J{\char95}1));\\
\indent \indent else\\
\indent \indent \indent Append({\char126}n2,\#(I2[s] diff J{\char95}1));\\
\indent \indent end if;\\
\indent end for;\\
//(C3)\\
\indent Min := Minimum(\{Integers()| n1[s]+n2[s]: s in [1..15]\});\\
\indent for i in [1..15] do\\
\indent \indent if (n1[i] + n2[i]) eq Min then\\
\indent \indent \indent s{\char95}0 := i; break;\\
\indent \indent end if;\\
\indent end for;\\
//(C4)\\
\indent tilde{\char95}G := G;\\
//(C5) and (C6)\\
\indent repeat\\
\indent \indent if \#(I1[s{\char95}0] meet J1(tilde{\char95}G)) lt 4 then\\
\indent \indent \indent calI1 := I1[s{\char95}0] meet J1(tilde{\char95}G);\\
\indent \indent else\\
\indent \indent \indent calI1 := I1[s{\char95}0] diff J1(tilde{\char95}G);\\
\indent \indent end if;\\
\indent \indent if (IsEmpty(calI1) eq false) then\\
\indent \indent \indent i{\char95}0 := Minimum(calI1);\\
\indent \indent \indent tilde{\char95}G := HorizontalJoin(tilde{\char95}G,h{\char95}vector(4,i{\char95}0));\\
\indent \indent end if;\\
\indent until IsEmpty(calI1) eq true;\\
//(C7) and (C8)\\
\indent repeat\\
\indent \indent if \#(I2[s{\char95}0] meet J1(tilde{\char95}G)) le 4 then\\
\indent \indent \indent calI2 := I2[s{\char95}0] meet J1(tilde{\char95}G);\\
\indent \indent else\\
\indent \indent \indent calI2 := I2[s{\char95}0] diff J1(tilde{\char95}G);\\
\indent \indent end if;\\
\indent \indent if (IsEmpty(calI2) eq false) then\\
\indent \indent \indent i{\char95}0 := Minimum(calI2);\\
\indent \indent \indent tilde{\char95}G := HorizontalJoin(tilde{\char95}G,h{\char95}vector(4,i{\char95}0));\\
\indent \indent end if;\\
\indent until IsEmpty(calI2) eq true;\\
\indent tilde{\char95}G := Matrix(BR, tilde{\char95}G);\\
\indent return tilde{\char95}G;\\
end function;\\
/*\\
SOconst{\char95}code{\char95}dim4(C) gives a   shortest  SO embedding   of an $[n,4]$ linear code.\\
Input: An $[n,4]$ code $C$\\
Output: A shortest SO embedding of $C$\\
*/\\
function SOconst{\char95}code{\char95}dim4(C)\\
\indent G := GeneratorMatrix(C);\\
\indent tilde{\char95}G := SOconst{\char95}matrix{\char95}dim4(G);\\
\indent tilde{\char95}C := LinearCode(tilde{\char95}G);\\
\indent return tilde{\char95}C;\\
end function;
}
\subsection{Algorithm~\ref{alg: SO embedding dim5} in MAGMA: Construction of an SO embedding  for higher dimensions }\label{apx: alg dim ge 5}
{\tt \scriptsize 
\noindent/* \\
SOconst{\char95}matrix{\char95}dim{\char95}ge5(G) gives a generator matrix for an SO embedding of an $[n,k]$ linear code for $k \ge 5$. \\
Input: A generator matrix $G$ of an $[n,5]$ code for $k \ge 5$\\
Output: A generator matrix $\widetilde{G}$ for an SO embedding\\
*/\\
function SOconst{\char95}matrix{\char95}dim{\char95}ge5(G)\\
\indent T :=[**];\\
\indent I{\char95}0 := [**];\\
\indent repeat\\
\indent \indent n := Ncols(G);\\
\indent \indent k := Nrows(G);\\
\indent \indent R := [**];\\
\indent \indent Ind := [1..k];\\
\indent \indent One := Matrix(GF(2),[[1]]);\\
\indent \indent Zero := Matrix(GF(2),[[0]]);\\
\indent \indent OOne := Matrix(GF(2),[[1,1]]);\\
\indent \indent ZOne := Matrix(GF(2),[[0,1]]);\\
\indent \indent ZZero := Matrix(GF(2),[[0,0]]);\\
\indent \indent for i in [1..k] do\\
\indent \indent \indent Append({\char126}R,RowSubmatrix(G,i,1));\\
\indent \indent end for;\\
\indent \indent i{\char95}0 := 0;\\
\indent \indent for i in [1..k] do\\
\indent \indent \indent discriminant := 0;\\
\indent \indent \indent for j in [1..k] do\\
\indent \indent \indent \indent if InnerProduct(R[i],R[j]) eq 1 then\\
\indent \indent \indent \indent \indent discriminant := discriminant + 1;\\
\indent \indent \indent \indent end if;\\
\indent \indent \indent end for;\\
\indent \indent \indent if discriminant eq 0 then\\
\indent \indent \indent \indent i{\char95}0 := i;\\
\indent \indent \indent \indent break;\\
\indent \indent \indent end if;\\
\indent \indent end for;\\
\indent \indent if i{\char95}0 ne 0 then\\
\indent \indent \indent Exclude({\char126}Ind,i{\char95}0);\\
\indent \indent \indent G := R[Ind[1]];\\
\indent \indent \indent Remove({\char126}Ind,1);\\
\indent \indent \indent for i in Ind do\\
\indent \indent \indent \indent G := VerticalJoin(G,R[i]);\\
\indent \indent \indent end for;\\
\indent \indent else\\
\indent \indent \indent Case := 0;\\
\indent \indent \indent i{\char95}0 := 1;\\
\indent \indent \indent for i in [1..k] do\\
\indent \indent \indent \indent if InnerProduct(R[i],R[i]) eq 1 then\\
\indent \indent \indent \indent \indent Case := 1;\\
\indent \indent \indent \indent \indent i{\char95}0 := i;\\
\indent \indent \indent \indent \indent break;\\
\indent \indent \indent \indent end if;\\
\indent \indent \indent end for;\\
\indent \indent \indent Exclude({\char126}Ind,i{\char95}0);\\
\indent \indent \indent if Case eq 1 then\\
\indent \indent \indent \indent for j in Ind do\\
\indent \indent \indent \indent \indent if InnerProduct(R[i{\char95}0],R[j]) eq 0 then\\
\indent \indent \indent \indent \indent \indent R[j] := HorizontalJoin(R[j],Zero);\\
\indent \indent \indent \indent \indent else\\
\indent \indent \indent \indent \indent \indent R[j] := HorizontalJoin(R[j],One);\\
\indent \indent \indent \indent \indent end if;\\
\indent \indent \indent \indent end for;\\
\indent \indent \indent \indent R[i{\char95}0] := HorizontalJoin(R[i{\char95}0],One);\\
\indent \indent \indent else\\
\indent \indent \indent \indent for j in Ind do\\
\indent \indent \indent \indent \indent if InnerProduct(R[i{\char95}0],R[j]) eq 0 then\\
\indent \indent \indent \indent \indent \indent R[j] := HorizontalJoin(R[j],ZZero);\\
\indent \indent \indent \indent \indent else\\
\indent \indent \indent \indent \indent \indent R[j] := HorizontalJoin(R[j],ZOne);\\
\indent \indent \indent \indent \indent end if;\\
\indent \indent \indent \indent end for;\\
\indent \indent \indent R[i{\char95}0] := HorizontalJoin(R[i{\char95}0],OOne);\\
\indent \indent \indent end if;	\\
\indent \indent \indent RemoveRow({\char126}G, i{\char95}0);\\
\indent \indent \indent G := R[Ind[1]];\\
\indent \indent \indent Remove({\char126}Ind,1);\\
\indent \indent \indent for i in Ind do\\
\indent \indent \indent \indent G := VerticalJoin(G,R[i]);\\
\indent \indent \indent end for;\\
\indent \indent end if;\\
\indent \indent Append({\char126}T, R[i{\char95}0]);\\
\indent \indent Append({\char126}I{\char95}0, i{\char95}0);\\
\indent until k eq 5;\\
\indent pre{\char95}G := SOconst{\char95}matrix{\char95}dim4(G);\\
\indent NC := Ncols(pre{\char95}G);\\
\indent for i in [1..\#T] do\\
\indent \indent for j in [1..NC-Ncols(T[i])] do\\
\indent \indent \indent T[i] := HorizontalJoin(T[i],Zero);\\
\indent \indent end for;\\
\indent end for;\\
\indent for i in [\#I{\char95}0..1 by -1] do\\
\indent \indent NR := Nrows(pre{\char95}G);\\
\indent \indent tG := ZeroMatrix(GF(2), 1, NC);\\
\indent \indent for j in [1..I{\char95}0[i]-1] do\\
\indent \indent \indent tG := VerticalJoin(tG,pre{\char95}G[j]);\\
\indent \indent end for;\\
\indent \indent tG := VerticalJoin(tG,T[i]);\\
\indent \indent for j in [I{\char95}0[i]..NR] do\\
\indent \indent \indent tG := VerticalJoin(tG,pre{\char95}G[j]);\\
\indent \indent end for;\\
\indent \indent RemoveRow({\char126}tG, 1);\\
\indent \indent pre{\char95}G := tG;\\
\indent end for;\\
\indent return tG;\\
end function;\\
/*\\
SOconst{\char95}code{\char95}dim{\char95}ge5(C) gives an SO embedding of an $[n,k]$ linear code for $k \ge 5$.\\
Input: An $[n,k]$ code $C$ for $k \ge 5$\\
Output: An SO embedding of $C$\\
*/\\
function SOconst{\char95}code{\char95}dim{\char95}ge5(C)\\
\indent G := GeneratorMatrix(C);\\
\indent tilde{\char95}G := SOconst{\char95}matrix{\char95}dim{\char95}ge5(G);\\
\indent tilde{\char95}C := LinearCode(tilde{\char95}G);\\
\indent return tilde{\char95}C;\\
end function;\\
}

\bibliographystyle{IEEEtran}
\bibliography{IEEEabrv,reference}

\end{document}